\documentclass[a4paper,twocolumn,11pt,accepted=2022-07-20]{quantumarticle}
\pdfoutput=1
\usepackage[utf8]{inputenc}
\usepackage[english]{babel}
\usepackage[T1]{fontenc}
\usepackage{amsmath}

\usepackage{tikz}
\usepackage{lipsum}

\usepackage{fullpage}
\usepackage[bookmarks]{hyperref}
\usepackage{amssymb}
\usepackage{amsthm}
\usepackage{graphicx,color,colordvi}
\usepackage{bbm}
\usepackage{cleveref}
\usepackage{stmaryrd}
\usepackage[blocks]{authblk}
\usepackage{dsfont}
\usepackage{mathtools}

\usepackage{centernot}

\usepackage{braket}
\usepackage{bm}

\usepackage{pgfplots}

\pgfplotsset{compat=1.17}

\def\Re{{\operatorname{Re}}}

\def\openone{\leavevmode\hbox{\small1\kern-3.8pt\normalsize1}}
\def\ch{{\cal H}}

\def\sa{\operatorname{sa}}

\DeclareMathOperator{\dd}{d}
\newcommand{\BraKet}[2]{\left\langle #1 \left| #2 \right.\right\rangle}
\newcommand{\one}{\mathds{1}}

\def\RR{\mathbb{R}}

\def\NN{\mathbb{N}}

\def\11{\mathbb{I}}

 \newtheorem{theorem}{Theorem}
 \newtheorem{lemma}{Lemma}
 \newtheorem{proposition}{Proposition}
 \newtheorem{corollary}{Corollary}

 \theoremstyle{definition}
 \newtheorem{definition}{Definition}
\newtheorem{example}{Example}

\newcommand{\tr}{\mathop{\rm Tr}\nolimits}

\newcommand{\spec}{{\rm sp}}

\usepackage{pst-node}
\usepackage{tikz-cd} 

\newcommand{\cA}{{\cal A}}
\newcommand{\cB}{{\cal B}}
\newcommand{\cD}{{\cal D}}
\newcommand{\cE}{{\cal E}}

\newcommand{\cT}{{\cal T}}
\newcommand{\cH}{{\cal H}}

\newcommand{\cJ}{{\cal J}}

\newcommand{\cL}{{\cal L}}

\newcommand{\cX}{{\cal X}}

\newcommand{\cM}{{\mathcal{M}}}

\newcommand{\bL}{\mathbf{L}}
\newcommand{\bA}{\mathbf{A}}
\newcommand{\be}{\mathbf{e}}
\newcommand{\bu}{\mathbf{u}}
\newcommand{\bX}{\mathbf{X}}
\newcommand{\bY}{\mathbf{Y}}

\def\d{\mathrm{d}}
\def\e{\mathrm{e}}

\renewcommand{\Im}{\operatorname{Im}}

\usepackage{graphicx}
\usepackage{setspace}
\usepackage{verbatim}
\usepackage{subfig}

\theoremstyle{definition}

 \theoremstyle{remark}
 \newtheorem{remark}{Remark}

\numberwithin{equation}{section}

\DeclareRobustCommand\openone{\leavevmode\hbox{\small1\normalsize\kern-.33em1}}

\newcommand{\id}{{\rm{id}}}
\newcommand{\beq}{\begin{equation}}
	\newcommand{\eeq}{\end{equation}}
\newcommand{\bea}{\begin{eqnarray}}
	\newcommand{\eea}{\end{eqnarray}}
\newcommand{\beas}{\begin{eqnarray*}}
	\newcommand{\eeas}{\end{eqnarray*}}

\title{Deviation bounds and concentration inequalities for quantum noises}

\author{Tristan Benoist}
\affiliation{ Institut de Mathématiques de Toulouse, UMR5219,
Université de Toulouse, CNRS, UPS, F-31062 Toulouse Cedex 9, France}
\author{Lisa H\"{a}nggli}
\author{Cambyse Rouz\'{e}}
\affiliation{Department of Mathematics, Technische Universit\"at M\"unchen, 85748 Garching, Germany}
\affiliation{Munich Center for Quantum Science and Technology (MCQST), M\"unchen, Germany}

\begin{document}

\maketitle
\onecolumn

\begin{abstract}
  We provide a stochastic interpretation of non-commutative Dirichlet forms in the context of quantum filtering. For stochastic processes motivated by quantum optics experiments, we derive an optimal finite time deviation bound expressed in terms of the non-commutative Dirichlet form. Introducing and developing new non-commutative functional inequalities, we deduce concentration inequalities for these processes. Examples satisfying our bounds include tensor products of quantum Markov semigroups as well as Gibbs samplers above a threshold temperature.
\end{abstract}

\section{Introduction}
Distinguishing unknown states of a given system constitutes one of the most fundamental tasks in experimental science, information theory and statistics -- in the classical, as well as in the quantum realm. Classically, the simplest formulation of the problem is as follows: An experimentalist is given independent samples $Y_1=y_1,\dots,Y_n=y_n$ from some unknown distribution $\pi$ taking values in a finite sample space $\Omega$ and is required to learn the distribution $\pi$. A natural estimate for the sought after distribution is then given by the empirical measure
\begin{align}\label{spatialLn}
L_n = \frac{1}{n}\sum_{i=1}^n \delta_{Y_i}\,,
\end{align}
where $\delta_{a}$ refers to the Dirac distribution at $a\in \mathbb{R}$. Noting that $L_n$ is a random element of the set $\cM_1(\Omega)$ of probability measures endowed with the topology of the total variation distance, we may quantify its aptitude for the intended task in terms of the probability of $L_n$ being close to the desired distribution $\pi$. More precisely, Sanov's theorem describes the asymptotic efficiency of the empirical measure as an estimator of the measure $\pi$: denoting by $\mathcal{B}(\mu,R)$ the open ball centered at $\mu\in\cM_1(\Omega)$ of radius $R$, the theorem states that for any $R> 0$,
\begin{equation}\label{sanov}
\lim_{n\to\infty} \frac{1}{n}\ln \mathbb{P}(L_n \in \mathcal{B}(\pi,R)^c) = - \inf_{\nu\in \mathcal{B}(\pi,R)^c} D(\nu\|\pi)\,.
\end{equation}
Here $D(\mu\|\nu)$ denotes the relative entropy between two probability measures $\mu<\!<\nu$, which is defined as $D(\mu\|\nu)=\sum_{i\in\Omega}\mu(i)\ln({\mu(i)}/{\nu(i)})$. In words, Sanov's theorem tells us that the probability of the empirical measure being at least distance $R$ away from the true measure $\pi$ decreases exponentially, with an asymptotic rate given by the right-hand side of \Cref{sanov}.

Recall that in the above first formulation of the problem of distinguishing unknown states, we assumed access to arbitrary samples drawn from the unknown distribution $\pi$ -- a fact employed in the theorem of Sanov.
% and we hence refer to the estimator $L_n$ as a spatial empirical measure.
However, in practice, the distribution $\pi$ is often generated from a random process (e.g.\ , a Gibbs sampler \cite{Levin2017}). In particular, $\pi$ may be the (unique) stationary measure of a time-continuous Markov process $t\mapsto X_t$ on $\Omega$, which corresponds to the continuous-time Markov semigroup $t\mapsto e^{tL}$ on $L^\infty(\Omega)$ generated by $L$. In this case, the measure $\pi$ is only reached as $t\to\infty$ and the experimentalist does not have access to the random variables $Y_j$ for $t<\infty$. It is thus crucial to find a dynamical strategy to learn $\pi$ from measurements performed at finite times. One way to achieve this goal is to continuously measure the sample paths corresponding to instances of the underlying process \cite{DV75}. In that case, the occupation time
\begin{align}\label{temporalLt}
L_t = \frac{1}{t}\int_0^t \delta_{X_s} \d s\,,
\end{align}
% which we also refer to as the temporal empirical average in order to distinguish it from the estimator $L_n$.
is a candidate for the estimator of the target measure $\pi$. Indeed, by the Donsker-Varadan theorem \cite{Hollander00}, we have that for any Borel set $B\subset \cM_1(\Omega)$,
\begin{align}\label{LDPclass1}
-   \inf_{\mu \in \operatorname{int}(B)} I(\mu)\le \liminf_{t\to\infty} \frac{1}{t}\ln \mathbb{P}(L_t \in B)\le \limsup_{t\to\infty} \frac{1}{t}\ln \mathbb{P}(L_t \in B) \le -   \inf_{\mu \in \operatorname{cl}(B)} I(\mu)\,,
\end{align}
with $\operatorname{int}(B)$ the interior of $B$, $\operatorname{cl}(B)$ the closure of $B$, and 
\[
I(\nu) := \sup_{u>0} \left[ -\sum_{i\in\Omega} \nu(i) \frac{(Lu)(i)}{u(i)} \right]\,,
\]
where the supremum runs over functions $u:\Omega\to(0,\infty)$.
% and for $O\subset \cM_1(\Omega)$ open,
% \begin{align}\label{LDPclass2}
% \liminf_{t\to\infty} \frac{1}{t}\ln \mathbb{P}(L_t \in O) \geq -  \inf_{\mu \in O} I(\mu)\,.
% \end{align}
Since the function $I(\nu)$ only vanishes for $\nu=\pi$ being the unique stationary measure of the Markov chain, Eq.~\eqref{LDPclass1} guarantees that the probability of the occupation time $L_t$ lying in a set not containing $\pi$ decays exponentially in $t$ with a rate determined by the function $I(\cdot)$. In other words, the probability of $L_t$ not constituting a good estimator for the target measure $\pi$ decays exponentially in time. Moreover, in the case of a reversible generator $L$, i.e.\ , $\pi(i) L_{ij} = \pi(j) L_{ji}$, we have that $I(\nu) =\mathcal{E}_L(f)$, where $f=\sqrt{{d\nu}/{d\pi}}$ and
\begin{align}\label{classDirich}
\cE_L(g):=-\mathbb{E}_\pi[g\,L(g)] \equiv -\sum_{i,j\in\Omega}\pi(i)\, g(i)\,L_{ij}\,g(j)\ ,\quad g:\Omega \to \mathbb{R}\ ,
\end{align}
is the so-called Dirichlet form associated with the generator $L$. More generally, Deuschel and Stroock proved in \cite[Theorem 5.3.10]{deuschel2001large} (see also \cite{Wu1994} for an extension to unbounded functions $f$ when $\Omega$ is a general Polish space) that for any function $f:\Omega\to\mathbb{R}$, all $r\ge 0$, and all initial distributions $\nu\in\cM_1(\Omega)$, it holds
\begin{align}\label{deuschelstroock}
    \lim_{t\to\infty}\frac{1}{t}\,\log\mathbb{P}_\nu\,\Big(\,\frac{1}{t}\int_0^t\,f(X_s)\,ds-\mathbb{E}_\pi[f]>r\Big)=-I_f\big(\mathbb{E}_\pi[f]+r\big)\,,
\end{align}
where $I_f$ is the lower semi-continuous regularization of the function
\begin{align*}
    J_f:\mathbb{R}\to\overline{\mathbb{R}},\quad x\mapsto \inf_{g:\Omega\to\mathbb{R}}\Big\{ -\mathbb{E}_\pi[\,g L(g)\,]\,\Big|\,\mathbb{E}_\pi[g^2]=1,\,\mathbb{E}_\pi[\,fg^2]=x \Big\}\,,
\end{align*}
with the convention that $\inf\emptyset=\infty$. 

Although they correspond to two different experimental setups, the spatial and temporal empirical averages $L_n$ and $L_t$ are intimately related through a functional inequality known as the \textit{logarithmic Sobolev inequality} (LSI): The generator $L$ is said to satisfy the latter if there exists a positive constant $\alpha$ such that for any probability measure $\nu<\!<\pi$, and $f=\sqrt{{d\nu}/{d\pi}}$,
\begin{align}\label{eq:LSI}
\alpha\,D(\nu\|\pi)\le\mathcal{E}_L(f)\,.
\end{align}
The largest constant $\alpha$ achieving this bound is called the \textit{logarithmic Sobolev constant} of $L$ and is denoted by $\alpha_2(L)$. As explained above, each side of \eqref{eq:LSI} corresponds to a rate function of a large deviation principle associated with two distinct strategies used to gain information on $\pi$. Thus, the LSI bears the statistical interpretation of providing a comparison between the asymptotic efficiency of the spatial and temporal empirical averages: up to the multiplicative constant $\alpha_2(L)$, the temporal strategy performs at least as well as the spatial strategy. From a more practical point of view, the logarithmic Sobolev inequality can also be used to derive finite time concentration bounds on the tail probability in \eqref{deuschelstroock} through a so-called \textit{transporation-information inequality} \cite{Wu00,Guillin2009}: assuming that $\Omega$ is a metric space with metric $d:\Omega\times\Omega\to\mathbb{R}_+$, the latter relates the Wasserstein distance between the measure $\pi$ and any other probability measure $\nu<\!<\pi$ to the Dirichlet form evaluated at $f=\sqrt{d\nu/d\pi}$:
\begin{align}\label{transportinfo}
    W_1(\nu,\pi)\le \,\sqrt{2C \,\cE_L(f)}\,.
\end{align}
Here $C>0$ is some constant independent of $\nu$, and 
$W_1(\nu,\mu):=\max_{\|f\|_{\operatorname{Lip}}\le 1}\mathbb{E}_\mu[f]-\mathbb{E}_\nu[f]$ is the Wasserstein distance between $\nu$ and $\pi$. The inequality \eqref{transportinfo} implies concentration bounds for the estimator $L_t$ of the following form \cite{Guillin2009}: For any Lipschitz function $f:\Omega\to\mathbb{R}$, any initial measure $\nu<\!<\pi$, and all $r>0$, 
\begin{align}
    \mathbb{P}_\nu\Big(\,\frac{1}{t}\int_0^t\,f(X_s)\,ds-\mathbb{E}_\pi[f]>r\Big)\le \,\Big\|\frac{d\nu}{d\pi}\Big\|_{L^2(\pi)}\,\exp\Big(-\frac{tr^2}{2C\,\|f\|_{\operatorname{Lip}}^2}\Big)\,.
\end{align}

In a typical quantum optics experiment \cite{Wiseman2009}, the collected data is represented by some quantum noises. In absence of a system perturbing the electromagnetic field, say an atom, the measured time resolved signal has either the law of a Brownian motion or of a `zero intensity' Poisson process. In the presence of the atom, the Brownian motion gains a drift depending on the atom's state. Similarly, the Poisson process gains a positive intensity which depends also on the atom's state. In this article, we focus our study precisely on these time resolved signal processes. We therefore do not derive bounds on quantities generalizing expressions such as $\int_0^tf(X_s)\d s$ as introduced before, but rather consider random processes $X_t$ which either satisfy a stochastic differential equation of the form $\d X_t= a_t \d t +\d W_t$, where $t\mapsto a_t$ is predictable and $t\mapsto W_t$ a Brownian motion, or such as $X_t$ is a Poisson process with stochastic predictable intensity. These choices have two motivations: first they correspond to the actual measurement processes occurring in current experiments and second it is for these processes that we can obtain deviation bounds from non-commutative generalizations of Dirichlet forms.

\paragraph{Quantum setup}

Let us present our framework in more details. Without conditioning on the measured data, the evolution of the state of the atom is modeled by a quantum Markov semigroup (QMS), that is a semigroup $t\mapsto e^{t\cL}$ of completely positive, unital maps acting on the algebra $\cB(\cH)$ of linear operators on a finite dimensional Hilbert space $\cH$. These are generalizations of Markov semigroups, changing the vector space from positive vectors to positive semi-definite matrices. A state of the atom is then lifted from a $\ell^1$ normalized positive vector in $\mathbb R^d$ to a trace one positive semi-definite matrix in $\cB(\cH)\equiv M_d(\mathbb C)$:
 $$\mathcal D(\cH)=\{\rho\in M_d(\mathbb C): \rho\geq 0, \tr \rho=1\}.$$
 
 We assume that $t\mapsto e^{t\cL}$ is primitive, which means that it possesses a unique full-rank invariant state $\sigma$. As recalled below (cf. \Cref{thm_Lindblad}), the generator $\cL$ admits a Lindblad form: there exist $H\in\cB_{\operatorname{sa}}(\cH)$, $k\in\mathbb{N}$ and $\mathbf{L}:\cH\to\cH\otimes \mathbb{C}^k$ such that
\begin{align}
 \forall X\in\cB(\cH),\quad    \mathcal{L}(X)= i[H,X]+\mathbf {L}^*(X\otimes \id_{\mathbb C^k})\mathbf {L}-\frac{1}{2}(\mathbf{L}^*.\mathbf{L}\,X+X\,\mathbf{L}^*.\mathbf{L})\,.
\end{align}
In analogy with the classical setting, the Dirichlet form associated to the generator $\cL$ is defined by
\begin{align}
\cE_\cL(X):= \langle X,\,\cL(X)\rangle_\sigma\,,\qquad X\in\cB(\cH)\,,
\end{align}
where $(X,Y)\mapsto \langle X,Y\rangle_\sigma:=\tr[\sigma^{\frac{1}{2}}X\sigma^{\frac{1}{2}}Y]$ is the so-called KMS inner product corresponding to the state $\sigma$, of associated norm $\|.\|_{L^2(\sigma)}$. 

The connection between the QMS and the measured signal is made through a dilation of the former. The semigroup $t\mapsto e^{t\cL}$ admits a unitary dilation $t\mapsto U_t$ modelling the continuous time dynamics resulting from the interaction of the quantum system and its environment. The family of unitaries $t\mapsto U_t$ satisfies a quantum stochastic differential equation (QSDE) recalled in \Cref{QSDE}. The processes modeling the time resolved measurement are derived from the quantum noises describing the environment in such QSDE. The operators $\mathbf L$ encode the interaction between the atom and the environment and prescribe what are the induced drifts and intensities of the signals.

We are interested in estimating the expectation values, with respect to the state $\sigma$, of linear combinations of $L_j+L_j^*$ and $L_j^*L_j$, with $L_j$ the components of the vector $\mathbf{L}$. For example,
\begin{align}
    \langle O(u)\rangle_\sigma:=\tr[\,\sigma \,O(u)]\,,\qquad O(u):=\sum_{j=1}^ku_j\,(L_j+L_j^*)\,,
\end{align}
where $u$ is a normalized vector in $\mathbb{R}^k$. These expectations are accessed through the above mentioned quantum noise measurements. Considering the Brownian motion with an added drift, the time averaged measured signal, and therefore our estimator, is given by
\begin{align}\label{qestimator}
L_t(u)= \frac{1}{t}\,\int_0^t \tr[O(u)\varrho_{s-}]\d s\,+\frac{W_t}{t} \,,
\end{align}
where $t\mapsto W_t$ is a standard Brownian motion. Here $t\mapsto \varrho_t$ is a stochastic process on the system's states initiated at the deterministic state $\rho\in\cD(\cH)$, whose average evolution equals $\mathbb{E}[\varrho_t]=e^{t\cL^*}(\rho)$. The fact that the drift is not deterministic but merely an adapted process is a prescription from quantum physics. The average estimator converges to the mean of $O(u)$ in the stationary state $\sigma$ since
\begin{equation}
\mathbb E[L_t(u)]=\frac{1}{t}\,\tr\big[\,O(u)\int_0^t\e^{s\mathcal L^*}(\rho)\,\d s\big]\underset{t\to\infty}{\longrightarrow} \,\langle O(u)\rangle_\sigma\,.
\end{equation}
Similarly, for Poisson processes the estimator is  $R_t(u)=\frac1t N_t(u)$, with $N_t(u)$ a Poisson process with intensity $t\mapsto \tr[(u.\mathbf L)^*u.\mathbf L \varrho_{t}]$. Hence $\mathbb E[R_t(u)]$ converges to $\tr[(u.\mathbf L)^*u.\mathbf L\sigma]$.

In \cite{kummerer_pathwise_2004} it was proved that the estimators $L_t$ and $R_t$ converge also almost surely.
Our main goal is to refine this result; in particular, we seek for deviation bounds on the measurement signals in the presence of an atom. Our main tools are Dirichlet forms and functional inequalities. 

% we couple our system with bosonic environments playing the role of our measurement devices. This results in a modification of the unitary dilation $U_t$

\paragraph{Summary of main results}
Sanov's theorem admits a quantum generalization \cite{BDKSSS}, and Deuschel-Stroock's large deviation theorem was also extended to the quantum discrete and continuous time settings \cite{JPW,vanHorssen2015} -- however, the link to the quantum Dirichlet form was not made in these articles. The first goal of this paper is to fill this missing gap. More precisely, we aim at quantifying the probability that the estimator $L_t(u)$ deviates from the average $\langle O(u)\rangle_\sigma$. We first prove in \Cref{thm:mainbound} the following upper bound: for any $r>0$,
\begin{align}
    \mathbb{P}(L_t(u)-\langle O(u)\rangle_\sigma\ge r)\le \|\sigma^{-\frac{1}{2}}\rho\sigma^{-\frac{1}{2}}\|_{L^2(\sigma)}\exp\Big(-t \inf_{\|X\|_{L^2(\sigma)}=1}\cE_\cL(X)+\frac{1}{2}\,(r+\langle O(u)\rangle_\sigma-f_u(X))^2\Big)
\end{align}
for some explicit function $f_u:\cB(\cH)\to \mathbb{R}$. \Cref{thm:mainbound} is a non commutative version of the commutative finite time bound from \cite[Theorem 1]{Wu00}. As for the commutative bound, it is optimal, in the sense that, when the generator $\cL$ is symmetric with respect to the KMS inner product $\langle .,.\rangle_\sigma$, the exponential rate is the rate function of the large deviation principle verified by the estimator $L_t(u)$  (see \Cref{LDP} and \cite[Theorem 4.2.58]{deuschel2001large} for the commutative version): for any Borel set $B\subset \mathbb R$,
\begin{align}\label{LDPquant}
-\inf_{{s}\in \operatorname{int}(B)}I_u({s})\leq \liminf_{t\to\infty}\frac1t\log \mathbb P(L_t(u)\in B)\leq\limsup_{t\to\infty}\frac1t\log \mathbb P(L_t(u)\in B)\leq-\inf_{{s}\in \operatorname{cl}(B)}I_u({s})\,,
\end{align}
with $\operatorname{int}(B)$ the interior of $B$, $\operatorname{cl}(B)$ the closure of $B$, and where
\begin{align}
    I_u(s):= \inf_{\|X\|_{L^2(\sigma)=1}}\,\cE_\cL(X)+\frac{1}{2}\,(s-f_u(X))^2\,.
\end{align}

In analogy with the classical setting, the generator $\cL$ is said to satisfy a quantum logarithmic Sobolev inequality (qLSI) if there exists a constant $\alpha>0$ such that, for any $\rho\in\cD(\cH)$ \cite{[OZ99],bardet2018hypercontractivity},
\begin{align}\label{qLSI}
    \alpha\,D(\rho\|\sigma)\,\le\, \cE_\cL(\sigma^{-\frac{1}{4}}\,\sqrt{\rho}\,\sigma^{-\frac{1}{4}})\,.
\end{align}
Here, we leave the problem of finding a statistical interpretation of \Cref{qLSI} open. Instead, we explore in \Cref{sec:concentration} the use of qLSI as well as other functional inequalities in deriving finite time concentration bounds on the tail probability of the random variable $L_t(u)-\langle O(u)\rangle_\sigma$. To serve this purpose, we consider a quantum generalization of the Wasserstein distance of order $1$ \cite{rouze2019concentration}: assuming that the generator $\cL$ is symmetric with respect to the so-called GNS inner product $(X,Y)\mapsto \tr[\sigma X^*Y]$, there exist parameters $\omega_i\in\mathbb{R}$, $i\in[k]$, such that \cite{fagnola2007generators} for all $ i\in[k]$, $   \sigma\, L_i\,\sigma^{-1}=e^{-\omega_i}L_i$.
 Then, the quantum Wasserstein distance associated to the generator $\cL$  between two quantum states $\rho_1$ and $\rho_2$ is defined as
\begin{align}
    W_{1,\cL}(\rho_1,\rho_2):=\sup_{\substack{X=X^*\\\|X\|_{\operatorname{Lip}}\le 1}}\,\tr[\rho_1 X]-\tr[\rho_2 X]\,,
\end{align}
where the non-commutative Lipschitz constant of an observable $X=X^*$ is defined as 
 $$\|X\|_{\operatorname{Lip}}:=\Big(\sum_{j\in [k]}  (\e^{-\omega_j/2}+\e^{\omega_j/2})\,\|[L_j,X]\|_\infty^2 \Big)^{1/2}\,.$$
 Extending the commutative framework of \cite{Guillin2009}, we say that the generator $\cL$ satisfies a \textit{non-commutative transportation-information} inequality of constant $C>0$ if for any state $\rho\in\cD(\cH)$, 
 \begin{align}\label{qTIintro}
     W_{1,\cL}(\rho,\sigma)\le \sqrt{2C\,\cE_\cL(\sigma^{-\frac{1}{4}}\,\sqrt{\rho}\,\sigma^{-\frac{1}{4}})}\,.
 \end{align}
 Building on previous results from  \cite{carlen2017gradient,rouze2019concentration,Datta2020}, we prove in \Cref{sec:functional} that \eqref{qLSI} implies the transportation-information inequality \eqref{qTIintro} with constant $C=\alpha^{-2}8^{-1}$. The latter is then proven to imply a Gaussian concentration inequality for the trajectory of the following form in \Cref{sec:concentration1}: there exist an explicit $u$ dependent constant $C_u$ such that for all $t,r>0$ and any initial state $\rho\in\cD(\cH)$,
\begin{align*}
\mathbb{P}\Big(L_t(u) -\langle O(u)\rangle_\sigma>r\Big) \le     \|\Gamma_\sigma^{-1}(\rho)\|_{\mathbb{L}_{2}(\sigma)}\, \exp\left(-\frac{t\,r^2}{C_u\|\Delta_\sigma^{\frac{1}{4}}(L^*_u)+\Delta_\sigma^{-\frac{1}{4}}(L_u)\|_{\operatorname{Lip}}^2}\right)\,.
\end{align*}
This bound is to our knowledge the first Gaussian concentration bound derived for the estimators $L_t(u)$ in the non commutative setting.

\paragraph{Layout of the paper}
\Cref{secpreliminaries} introduces the concepts that we are going to use in the rest of the article, namely quantum stochastic differential equations (\Cref{QSDE}), quantum trajectories (\Cref{qtraj}) and quantum Markov semigroups (\Cref{subsec:qms}). In \Cref{sec:tail}, we provide our main upper bound (\Cref{thm:mainbound}) and use it to derive a large deviation principle in the case of a reversible semigroup (\Cref{LDP}). The proof of \Cref{thm:mainbound} is postponed to \Cref{sec:proofmain}. In \Cref{sec:concentration}, we provide a self-contained introduction to classical and quantum transportation cost metrics (\Cref{Wassersteindist}) and related inequalities (\Cref{sec:functional}), and apply the latter to the derivation of concentration inequalities in \Cref{theo:ti}. We end this article with a few examples illustrating our bounds, including generalized depolarizing channels and tensor products of quantum Markov semigroups, in \Cref{examples}.

% In contrast with the classical setting of \Cref{LDPclass1,LDPclass2}, the exponents in \Cref{LDPquant} are expressed in terms of a further optimization over the parameter $s$, which is directly related to the Brownian term in the expression for the estimator \eqref{qestimator} arising from the indirect measurement. However, the quantum logarithmic Sobolev inequality can still be used to relate the large deviation exponent for $L_t(u)$ to the one from the quantum Sanov theorem \textcolor{red}{for quantum Donsker Varadan we sup over all $X$, not those in a given ball. This is because we fix an observable here, whereas for Sanov we learn the states (max over all bounded observables). The right thing to compare would be to the relative entropy of the measure in the spectrum of $O(u)$ in the state $\sigma$, or the observables for occupation times like classically. Or we need a $p_t$ analogue of the $p_n$ in the quantum Sanov. }
% \textcolor{blue}{Easy solution: say it is not clear how LSI can be used to relate the two LDP, it can be used to get concentration bounds. More precisely, it implies the transportation information inequality which we introduce here, blabla}

\section{Notations and preliminaries}\label{secpreliminaries}
The set of $k\times k$ matrices with complex entries is denoted by $M_k(\mathbb{C})$, the unit ball in $\mathbb{C}^k$ by $S^{k-1}(\mathbb{C})$, and the one in $\mathbb{R}^k$ by $S^{k-1}(\mathbb{R})$. Let $\cH$ be a separable Hilbert space. We denote by $\mathcal{B}(\cH)$ the space of linear operators on $\cH$,  by $\mathcal{B}_{\operatorname{sa}}(\cH)$ the subspace of self-adjoint operators on $\cH$, and by $\mathcal{B}_+(\cH)$ the cone of positive semidefinite operators. The adjoint of an operator $Y$ is written as $Y^*$. The operator norm on $\cB(\cH)$ is denoted by $\|.\|_\infty$, whereas Schatten norms are denoted by $\|.\|_p$ for $p\ge 1$. The identity operator on a vector space $\cA$ is denoted by $\id_\cA$, dropping the index
$\cA$ when the corresponding space is obvious from the context. We denote by $\mathcal{D}(\cH)$ the set of positive semidefinite, trace one operators on $\cH$, also called \emph{density operators}, and by $\cD_>(\cH)$ the subset of faithful density operators. In the following, we will often identify a density operator $\rho\in\mathcal{D}(\cH)$ and the \emph{state} it defines, that is, the positive linear functional $\mathcal{B}(\cH)\ni X\mapsto\tr[\rho \,X]$. Given a state $\rho\in\cD(\cH)$, a self-adjoint operator $O$ on $\cH$ and a real number $r\in\mathbb{R}$, we denote the probability that a measurement of $O$ in the state $\rho$ yields a value $\lambda >r$ by
\begin{align}
    \mathbb{P}_{\rho}\big(O>r\big):=\tr\big[\rho\,\one_{]r,\infty[}\,(O)\big]\,,
\end{align}
where $\one_{]r,\infty[}$ denotes the characteristic function of the interval $]r,\infty[$.

\subsection{Non-commutative noises and quantum stochastic calculus}\label{QSDE}
In this section we give a concise introduction to quantum noises and quantum stochastic calculus without proofs. For a complete introduction and proofs of the results stated here, we refer the reader to the pioneering work of Hudson and Parthasarathy \cite{PH84} (see also \cite{Parthasarathy} and references therein), but also to the book of Holevo \cite{Holevo2001}, and the introduction to quantum filtering by Bouten\emph{ et al.} \cite{Bouten2007}.

\subsubsection{Non-commutative noises}\label{sec:noncommnoises}
In non-commutative probability, a measurable space is replaced by a unital $*$-algebra of operators. The elements of the algebra replace random variables. Expectations are replaced by normalized positive linear forms called states. A standard non-commutative equivalent of the (commutative) canonical Brownian motion and Poisson process are defined by elements of a $*$-algebra of operators on a symmetric Fock space. The Wiener and Poisson measures are translated to the vacuum state of this Fock space. We now properly define these objects.

Fix a natural number $k$. Let $\Gamma_s^{(0)}=\mathbb C$, and for any natural number $n$, let
$$\Gamma_s^{(n)}=S_n L^2(\mathbb R_+;\mathbb C^k)^{\otimes n},$$
where $L^2(\mathbb{R}_+;\mathbb{C}^k)$ denotes the space of square integrable functions of the half-line with values in $\mathbb{C}^k$, and $S_n$ is the orthogonal projection onto the symmetric subspace of the tensor product. For example, $S_2f\otimes g=\frac12(f\otimes g+g\otimes f)$. The symmetric Fock space, $ \Gamma_s$, is then defined as the completion of the direct sum of the Hilbert spaces $\Gamma_s^{(n)}$:
$$\Gamma_s=\overline{\bigoplus_{n\in\mathbb N_0} \Gamma_s^{(n)}}.$$
Since we work only with symmetric Fock spaces, we will omit the index $s$ from now on. We will mostly work in reference to the vacuum state $\Omega\in\Gamma$, defined by $\Omega=1\oplus_{n\in\mathbb N}0$.

On the Fock space introduced above we now define creation and annihilation operators: For any $f\in L^2(\mathbb R_+;\mathbb C^k)$, the creation operator $a^*(f)$ is defined by
$$a^*(f)\Psi=\sqrt{n+1}S_{n+1}\, f\otimes \Psi,\quad \forall \Psi\in \Gamma^{(n)}.$$
The annihilation operator $a(f)$ is defined by
$$a(f)\Psi=\sqrt{n}\langle f|\otimes \id_{\Gamma_s^{(n-1)}} \Psi, \quad \forall \Psi\in \Gamma^{(n)},$$
where $\langle f\vert$ is the dual of $f$ in $L^2\equiv L^2(\mathbb{R}_+;\mathbb{C}^k)$.
These operators are extended by linearity to the space of finitely many particle vectors: $\Gamma^{\mathrm{fin}}=\{ \Psi\in\Gamma: \exists N\in\NN\mbox{ s.t. } \Psi\in \left(\bigoplus_{n=0}^N\Gamma^{(n)}\right)\oplus \left(\bigoplus_{n>N}0_{\Gamma^{(n)}}\right)\}.$ That way they are densely defined on $\Gamma$ since $\Gamma=\overline{\Gamma^{\mathrm{fin}}}$.
Computing the dual of $a^*(f)$ shows that $a^*(f)=a(f)^*$. It also follows from the definition that the creation and annihilation operators verify the canonical commutation relations,
\begin{align}\label{CCR}
[a(f),a^*(g)]=\langle f,g\rangle\,\id\ ,
\end{align}
where $\langle .,.\rangle$ denotes the inner product in $L^2(\mathbb{R}_+;\mathbb{C}^k)$. We now introduce the second quantization operator $\d\Gamma$ that appears in the non-commutative generalization of the Poisson process. For any bounded operator $h:L^2(\mathbb R_+;\mathbb C^k)\to L^2(\mathbb R_+;\mathbb C^k)$, let
$$\d\Gamma(h)=\bigoplus_{n\in\mathbb N_0} \d\Gamma_n(h),\quad \mbox{with} \quad   \d\Gamma_n(h)=\sum_{j=1}^n \id_{L^2}^{\otimes (j-1)}\otimes h \otimes \id_{L^2}^{\otimes (n-j)}$$
and domain $\Gamma^{\mathrm{fin}}$.
A standard example of such an operator is the particle number operator: for a fixed $p\in \mathbb N_0$ and $\Psi\in \Gamma^{(p)}\oplus 0\subset \Gamma$, we have $\d\Gamma(\id_{L^2})\Psi=p\Psi$. Note that $\d\Gamma$ is a morphism of the vector space of bounded operators on $L^2$: $\d\Gamma(h_1+h_2)=\d\Gamma(h_1)+\d\Gamma(h_2)$.

The creation and annihilation operators lead to the definition of the non-commutative replacement for the Brownian motion, which is based on the operator-valued functions
$$A_u:t\mapsto A_u(t)=a(\one_{[0,t[} \,u),\quad\mbox{and}\quad A_u^*:t\mapsto A_u^*(t)=a^*(\one_{[0,t[}\,u),$$
called the annihilation and creation processes respectively, with $u$ a unit vector in $\mathbb C^k$ and $\one_{[0,t[}$ the characteristic function of the interval $[0,t[$.
Note that by \Cref{CCR},
\begin{align}\label{CCRA}
[A_u^*(t),A_v(s)]=\langle v,u\rangle\min(s,t)\,\id ~~~~ \text{ while }~~~~[A_u(t),A_v(s)]=0\,,
\end{align}
where, by a slight abuse of notations, we used the same brackets $\langle.,.\rangle$ to denote the canonical inner product on $\mathbb{C}^k$. 
To simplify the notations, for an orthonormal basis $\mathbf{u}=\{u_1,\dotsc,u_k\}$ of $\mathbb C^k$  we write
$$\mathbf{A}_{\mathbf u}(t)=(A_{u_1}(t),\dotsc,A_{u_k}(t))\quad \mbox{and}\quad\mathbf{A}_{\mathbf{u}}^*(t)=(A_{u_1}^*(t),\dotsc,A_{u_k}^*(t))\,.$$
Using the second quantization operator we define gauge processes that relate to classical Poisson processes. For any $r\in M_k(\mathbb C)$, let
$$\Lambda_{r}:t\in\mathbb R_+\mapsto \d\Gamma(\operatorname{mult}(\one_{[0,t[}r)),$$
with $\operatorname{mult}(g)$ the operator on $L^2(\mathbb R_+;\mathbb C^k)$ of point-wise multiplication by the $k\times k$ matrix valued bounded function $g$. For a unit vector $u\in \mathbb C^k$ we denote $\Lambda_u=\Lambda_{|u\rangle\langle u|}$, and
for $\mathbf{u}$ an orthonormal basis of $\mathbb C^k$, we write
$$\Lambda_{\mathbf{u}\mathbf{u^*}}(t)=(\Lambda_{|u_i\rangle\langle u_j|}(t))_{i,j=1}^k$$
and $\mathbf{\Lambda}_{\mathbf{u}}(t)=(\Lambda_{u_1}(t),\dotsc,\Lambda_{u_k}(t))$.

Using these creation, annihilation and gauge processes, we can reconstruct the classical Brownian motion and Poisson processes. First note that, for $\mathbf{u}$ an orthonormal basis of $\mathbb C^k$, the set $\{ A_{u_j}(t)+A_{u_j}^*(t): t\in\mathbb R_+, j\in\{1,\dotsc,k\}\}$ is formed of commuting self-adjoint operators by \Cref{CCRA}. Hence, by the spectral theorem, the latter can be jointly represented as multiplication operators on an appropriate Hilbert space. Moreover, these multiplication operators in fact coincide with multiplication operators by $k$ independent Brownian motions. Similarly, the processes $t\mapsto \Lambda_{u_i}(t)+\sqrt{\lambda_i}A_{u_i}(t)+\sqrt{\lambda_i}A_{u_i}^*(t)+\lambda_i t$ can be jointly represented as $k$ independent Poisson processes of respective intensity $\lambda_i$. Let us now state this fact rigorously: 

\begin{theorem}\label{thm:mult_by_B}
Let $q\in \{0,\dotsc,k\}$. Let $B_1,\dotsc, B_q$ be $q$ independent Brownian motions and $N_{q+1},\dotsc, N_k$ be $p$ independent Poisson processes of respective intensities $\lambda_{q+1},\dotsc, \lambda_k$. Assume the Brownian motions and the Poisson processes are independent. Let $\mathcal W$ be the Hilbert space of square integrable, complex functions of $(B_1,\dotsc,B_q,N_{q+1},\dotsc,N_k)$.
Then, for any orthonormal basis $\mathbf{u}$ of $\mathbb C^k$, there exists a unitary transformation $J:\Gamma\to \mathcal W$, such that for any $t\in \mathbb R_+$, $i\in\{1,\dotsc,q\}$, and $j\in\{q+1,\dotsc,k\}$, 
$$A_{u_i}(t)+A_{u_i}^*(t)=J^*M(B_i(t))J\quad\mbox{and}\quad \Lambda_{u_j}(t)+\sqrt{\lambda_j}(A_{u_j}(t)+A_{u_j}^*(t))+\lambda_j t=J^*M(N_j(t))J,$$
with $M(x)$ the operator of multiplication by $x\in \mathbb R$. Moreover, $J\Omega=\mathbf 1$.
\end{theorem}

In \cite[Section 3 and 4]{Holevo2001}, Holevo provides a proof of Theorem~\ref{thm:mult_by_B} relying on chaos expansions to construct explicitly the unitary transformation $J$. Another proof, using Stone's theorem, is outlined by Bouten \emph{et al.} in \cite[Section 4.1 and 4.2]{Bouten2007}. Note that the proof can also be carried out using Gelfand's representation theorem for commutative C$^*$-algebras.

As $\mathbf{A}_{\bf u}+\mathbf A_{\bf u}^*$ is isomorphic to the multiplication by $\mathbf B_{\bf u}$, we will only work in the Fock space $\Gamma$ from now on, and denote $\mathbf B_{\bf u}=\mathbf A_{\bf u}+\mathbf A_{\bf u}^*$. Similarly, we denote $N_{u_j}=\Lambda_{u_j}(t)+\sqrt{\lambda_j}(A_{u_j}+A_{u_j}^*)+\lambda_j t$. The unitary transformation between the Hilbert spaces $\Gamma$ and $\mathcal W$ will always be implicit.

It is however essential to keep in mind that for $u, v\in S^{k-1}(\mathbb C)$ non-orthogonal, $[(A_u(t)+A_u^*(t)),(A_v(t)+A_v^*(t)]\neq 0$ and $[(\Lambda_u(t) + A_u(t)+A_u^*(t)+t),(\Lambda_v(t)+A_v(t)+A_v^*(t)+t)]\neq 0$.

Summarizing, using non-commutative noises, we constructed a $q$ dimensional Brownian motion $\mathbf B$ and $p$ independent Poisson processes $\mathbf N$ as multiplication operators on the symmetric Fock space $\Gamma$. In the next subsection, we start from these processes and introduce the so-called non-commutative Girsanov transform to unravel the family of stochastic processes we are concerned with in this paper.
 
\subsubsection{Non-commutative Girsanov transforms}

For classical stochastic processes, the drift of a diffusion or the intensity of a point process can be modified through a so-called Girsanov transformation: Let $(\mathbf{X},\mathbf{N})$ be a couple of processes where is $\bf X$ the solution of $d\mathbf{X}_t=\mathbf{a_t}\d t+\d \mathbf{W}_t$ with $\mathbf{W}$ a multidimensional Brownian motion and $\mathbf{N}$ is a vector of Poisson processes with corresponding vector of stochasic intensities $\mathbf{\lambda}_t$. Let $\mathbb P_T$ denote the law of this couple of processes up to time $T$. Let $\mathbb Q_T$ be a law such that, up to time $T$, $(\mathbf{X},\mathbf{N})$ is a couple of independent processes with $\mathbf{X}$ a multidimensional Brownian motion and $\mathbf{N}$ a vector of independent Poisson processes of unit intensity also independent of $\mathbf{X}$. Then, by Girsanovs theorem, there exists a $\mathbb Q$-martingale $Z$, such that $\d\mathbb P_T=Z_T\d\mathbb Q_T$. More precisely, $Z$ is a solution of
$$Z_t^{-1}\d Z_t= \mathbf{a}_t.\d \mathbf{X}_t+(\mathbf{\lambda}_t-\mathbf{1}).(\d\mathbf{N}_t -\d t),\quad Z_0=1$$
where $\mathbf{x}.\mathbf{y}$ is the dot product between $\mathbf{x}$ and $\mathbf{y}$.

We will now explore a non-commutative construction of such a law transformation. In the next subsection we provide a martingale $Z$ that performs said non-commutative law transformation, but requires the introduction of an auxiliary process. 

The transforms we discuss are essentially motivated by quantum physics. For example, the noises $\mathbf A$ and $\mathbf \Lambda$ can model an electromagnetic field. The transform then defines the state of the field (\emph{i.e.}\ , over the $*$-algebra of operators generated by $\mathbf A$ and $\mathbf\Lambda$) after its interaction with an atom. For a physical discussion of these models we refer the reader to \cite{Wiseman2009}. The transformation we consider can be obtained as limits of quantum physics axiomatic Hamiltonian dynamics (see \cite{attal_repeated_2006,derezinski_extended_2008}).

Our non-commutative transforms are defined using quantum stochastic calculus. The latter is a well established theory that makes sense of integrals like $\int_0^t \mathbf{F}_s.\d \mathbf{A}_{\mathbf{u}}(s)+\int_0^t \mathbf{G}_s.\d \mathbf{\Lambda}_{\mathbf {uu}^*}(s)$, for $s\mapsto \mathbf{F}_s$ and $s\mapsto\mathbf{G}_s$ appropriate adapted operator valued functions. They are adapted in the sense that for any $t\in \RR_+$, through the natural isomophism $\cH\otimes\Gamma\equiv \cH\otimes \Gamma_s(L^2([0,t[))\otimes \Gamma_s(L^2([t,\infty[))$, $\mathbf{F}_t$ is mapped to $\widehat{\mathbf{F}}_t\otimes \id_{\Gamma_s(L^2([t,\infty[))}$ with $\widehat{\mathbf{F}}_t$ an operator on $\cH\otimes \Gamma_s(L^2([0,t[))$, mimicking the classical stochastic calculus definition. We do not define such integrals further here; for an introduction to this subject we refer the reader to Parthasarathy's book \cite{Parthasarathy}. We use the definition of quantum stochastic integrals used there. For readers used to classical stochastic calculus, the difference can be summarized by a non-commutative version of It\^o rules: for $u,v\in S^{k-1}(\mathbb C)$ and $r_1,r_2\in M_k(\mathbb C)$,
\begin{equation}\label{eq:noncommito}
\begin{aligned}
\d A_{u}(t)\d A_{v}^*(t)=\langle u,v\rangle\d t\ ,&\quad \d \Lambda_{r_1}(t)\d A_{u}^*(t)=\d A_{r_1u}^*(t)\ ,\\ \d A_{u}(t)\d\Lambda_{r_1}(t)=\d A_{r_1^{\mathsf{T}}u}(t)\ ,&\quad \d \Lambda_{r_1}(t)\d \Lambda_{r_2}(t)=\d\Lambda_{r_1r_2}(t)\ ,
\end{aligned}
\end{equation}
and all other products of infinitesimal increments being $0$. Setting $\mathbb E_0$ as the vacuum state 
$$\mathbb E_0:\mathrm{Poly}(a(f),a^*(f), \d\Gamma(h))_{f\in L^2,h\in\mathcal B(L^2)}\to \mathbb C, A\mapsto \langle \Omega, A\Omega\rangle,$$
we have $\mathbb E_0(\d A_u(t))=\mathbb E_0(\d A_{u}^*(t))=\mathbb E_0(\d\Lambda_{r}(t))=0$ for any $u\in \mathbb C^k$ and $r\in M_k(\mathbb C)$. Anything we do in the present article involving quantum stochastic calculus relies only on these computational rules. They can thus be taken as axiomatic starting points.

To give an intuition for the stochastic integrals we will use, let us mention that for any bounded function $f\in L^2(\mathbb R_+)$, $u\in S^{k-1}(\mathbb C)$ and $r\in M_k(\mathbb C)$,
$$\int_0^t\overline{f(s)}\d A_u(s)=a(\one_{[0,t[}f u),\quad \int_0^t f(s)\d A_u^*(s)=a^*(\one_{[0,t[}f u),\quad\mbox{and}\quad \int_0^t f(s)\d\Lambda_{r}(s)=\d\Gamma(\operatorname{mult}(\one_{[0,t[}f\ r)).$$
%It follows from Theorem \ref{thm:mult_by_B} that, setting $q=k$, for a real continuous function $F\in L^2(\mathbb R_+;M_k(\mathbb C))$, and $\mathbf{u}$ an orthonormal basis of $\mathbb C^k$,
%$$\int_0^t F(s).(\d A_{\mathbf u}(t)+\d A_{\mathbf u}^*(t))=J^*M\left(\int_0^t F(s).\d \mathbf{B}(s)\right)J.$$

Let $\mathbf{e}$ be the canonical basis of $\mathbb C^k$, let $H\in \mathcal B(\mathcal H)$ be such that $H=H^*$, and let $\mathbf L\in \mathcal B(\mathcal H)\otimes \mathbb C^k$. Then we define $(s,t)\mapsto U_{t,s}$ to be the two parameter unitary group solution, for $t\geq s$, of
\begin{equation}\label{eq:def_U}
\d U_{t,s}= \left\{\left(-i H -\frac12 \mathbf L^*.\mathbf L\right)\d t + \mathbf{L}.\d\mathbf{A}_{\mathbf e}^*(t) - \d \mathbf{A}_{\mathbf{e}}(t).\mathbf{L}^* \right\}U_{t,s},\quad U_{s,s}=\id.
\end{equation}
Corollary 26.4 in \cite{Parthasarathy} ensures that a solution exists and is a unitary group. For $s=0$, we denote $U_{t,0}=U_t$. Note that For each $t\geq s>0$, $U_{t,s}$ is a unitary operator from $\mathcal H\otimes \Gamma$ to itself.

Physically, this group models the interaction between the atom described by $\mathcal H$ and the field described by the quantum noises. The law transformation we are concerned with is the one resulting from the unitary transformation of the initial state by the unitary $U_T$. Hence $U_T$ is the non-commutative version of Girsanov's transform.

Let us now introduce the filtered measurable space we work on:
\begin{definition}\label{def:probability space}
Let $\mathrm{D}$ be the set of c\`adl\`ag functions $\mathbb R_+\to \mathbb R^k$ equipped with the Skorhokhod topology. Let $\mathcal F$ be its Borel $\sigma$-algebra and $(\mathcal F_t)_{t\in\mathbb R_+}$ its usual filtration. Then $(\mathrm{D}, \mathcal F,(\mathcal F_t)_{t\in \mathbb R_+})$ is a filtered measurable space.
\end{definition}
As the definition of a probability measure relies on the spectral theorem for commutative von Neumann algebras, we need to map our unbounded noise operators to bounded ones. For that purpose, for any function $f\in \Omega$, let
\begin{align*}
\iota(f):\mathbb R_+&\to \mathbb C^k\\
t&\mapsto ((i+f_1(t))^{-1},\dotsc,(i+f_k(t))^{-1}).
\end{align*}
Note that $\iota$ is injective from $\mathrm D$ to the set of $\mathbb C^k$ valued functions of $\mathbb R_+$. It is moreover bounded with respect to the supremum norm. We can now define probability measure of our main concern:

\begin{definition}\label{def:P}
    Fix $q\in\{1,\dotsc,k\}$ and let $\mathbf{u}$ be an orthonormal basis of $\mathbb C^k$. Let $\rho\in \mathcal D(\mathcal H)$. For any $T\in \mathbb R_+$, let $\mathbb P_T$ be the probability measure on $(\mathrm{D},\mathcal F_T)$ pushed forward by $\iota^{-1}$ of the spectral measure of the smallest commutative von Neumann algebra containing
    $$(i+B_{u_i}(t))^{-1}\quad\mbox{and}\quad (i+\Lambda_{u_j}(t))^{-1},\qquad t\in[0,T[,\ i=1,\dotsc,q,\ j=q+1,\dotsc,k,$$
    with respect to $\Psi_T=U_T(\rho\otimes |\Omega\rangle\langle\Omega|)U_T^*$. Let $\mathbb P\equiv \mathbb{P}_{\rho\otimes\Omega}$ be the unique extension of $(\mathbb P_t)_{t\in \mathbb R_+}$ to $(\mathrm{D},\mathcal F)$. 
\end{definition}
The fact that $(\mathbb P_t)_{t\in \mathbb R_+}$ is a consistent family and thus can be extended to $\mathbb P$ by the Kolmogorov extension theorem is a standard result in quantum filtering; see \cite{Bouten2007}.

Note that Corollary 26.4 in \cite{Parthasarathy} holds for more general quantum stochastic differential equations than Eq.~\eqref{eq:def_U}. In particular, it can involve gauge processes $\Lambda_r$. However, we can restrict ourselves to Eq.~\eqref{eq:def_U} without loss of generality: According to the corollary, fixing $\mathbf S\in \cB(\mathcal H)\otimes M_k(\mathbb C)$ as a unitary operator, we could have used the unitary group $(s,t)\mapsto V_{t,s}$, which is the solution of
\begin{equation}\label{eq:alternative_U}
\d V_{t,s}= \left\{(-i H -\frac12 \mathbf L^*.\mathbf L)\d t + \mathbf{L}.\d\mathbf{A}_{\mathbf e}^*(t) - \d \mathbf{A}_{\mathbf{e}}(t).\mathbf{L}^*\mathbf{S} + \tr_{\mathbb C^k}[(\mathbf{S}-\id)\d\mathbf{\Lambda}_{\mathbf{ee}^*}(t)]\right\}V_{t,s},\quad V_{s,s}=\id,
\end{equation}
to define a reference state $\Phi_T$ that depends on the extra parameter ${\bf S}\in \mathcal B(\mathcal H)\otimes M_k(\mathbb C)$. However, using $\Omega\equiv \Omega_{\Gamma(L^2([0,t[))}\otimes \Omega_{\Gamma(L^2([t,+\infty[))}$, $A_{u}(t+\Delta t)-A_u(t)\equiv a(\one_{[t,t+\Delta t[})$, $\Lambda_r(t+\Delta t)-\Lambda_r(t)\equiv \d\Gamma( r\one_{[t,t+\Delta t[})$ and $a(\one_{[t,t+\Delta t[})\Omega_{\Gamma(L^2([t,+\infty[))}= \d\Gamma( r\one_{[t,t+\Delta t[})\Omega_{\Gamma(L^2([t,+\infty[))}=0$ for any $\Delta t>0$ and $t\geq 0$, we obtain that $\Phi_T$ and $\Psi_T$ are both solutions to the same differential equation. Moreover, since this differential equation has a unique solution, we get
$$\Phi_T=\Psi_T,\quad \forall T\geq0.$$
Therefore, the extra parameter $\bf S$ does not change the reference state $\Psi_T$. Hence, the measure $\mathbb P$ is independent of $\bf S$ and we set it to $\id$ in the remainder of the article.

The proof of the next proposition is standard in quantum filtering. It can be found in \cite{Bouten2007} for example. 
\begin{proposition}\label{prop:nc_girsanov}
    Fix $q\in\{1,\dotsc,k\}$ and let $\mathbf u$ be an orthonormal basis of $\mathbb C^k$. Let $\mathbf{X}$ and $\mathbf{Y}$ be the operator valued functions defined by
    $$X_i(t)=U_t^*B_{u_i}(t)U_t\quad\mbox{and}\quad  Y_j(t)=U_t^*\Lambda_{u_j}(t) U_t,\qquad t\in[0,T[,\ i=1,\dotsc,q,\ j=q+1,\dotsc,k,$$
    respectively.
    Then, $\mathbf{X}$ and $\mathbf{Y}$ are solutions of
    $$\d X_i(t)=U_t^* \big( L_{u_i}+L_{u_i}^*\big) U_t \d t +\d B_{u_i}(t),\quad X_i(0)=0,$$
    and
    $$ \d Y_j(t)=\d \Lambda_{u_j}(t)+ U_t^*L_{u_j}U_t\d A_{u_j}^*(t) +U_t^*L_{u_j}^*U_t\d A_{u_j}(t)+ U_t^* L_{u_j}^*L_{u_j} U_t \d t, \quad Y_j(0)=0,$$
    respectively, with $L_{u_i}\in \mathcal B(\mathcal H)$ defined by
    \begin{align}\label{eq:Lu}
     L_{u_i}=(\id_{\mathcal B(\mathcal H)}\otimes \langle u_i| )\mathbf L\,.
     \end{align}
    Furthermore, the measure $\mathbb P$ is the push forward by $\iota^{-1}$ of the spectral measure, with respect to $\rho\otimes |\Omega\rangle\langle\Omega|$, of the smallest commutative von Neumann algebra containing
    $$(i+X_i(t))^{-1}\quad\mbox{and}\quad (i+Y_j(t))^{-1},\qquad t\in \mathbb R_+,\ i=1,\dotsc,q,\ j=q+1,\dotsc,k.$$
\end{proposition}

In the present article we study the concentration properties of the measure $\mathbb P$, or, equivalently, the processes $\bf X$ and $\bf Y$. Note that for a fixed  ${\bf L}\in \mathcal B(\mathcal H)\otimes \mathbb C^k$ and initial state $\rho\in \mathcal D(\mathcal H)$, these two processes and their joint law $\mathbb P$ depend only on the choice of the unitary basis $\bf u$ and the index $q\in\{1,\dotsc,k\}$.

The fact that the set of operators $\{X_i(t), Y_j(t): i=1,\dotsc,q; j=q+1,\dotsc,k; t\in \mathbb{R}_+\}$ defines a commutative family is not obvious, but follows from the fact that $\bf u$ is an orthogonal basis and $U_{t+s}^*B_{u_i}(t)U_{t+s}=U_{t}^*B_{u_i}(t)U_{t}$ and $U_{t+s}^*\Lambda_{u_j}(t)U_{t+s}=U_{t}^*\Lambda_{u_j}(t)U_{t}$ for any $i,j$ and $t,s\in \mathbb R_+$ -- see \cite{Bouten2007} for example.

Whenever $\dim \mathcal H=1$, the operators $L_{u_j}$ are complex numbers, and from Theorem \ref{thm:mult_by_B} we obtain that $\bf X$ is a $q$-dimensional Brownian motion plus a deterministic drift linear in time, whereas $\bf Y$ is an independent $(k-q)$-tuple of independent Poisson processes with fixed intensity.

\subsection{Quantum trajectories and classical Girsanov's transform}\label{qtraj}
The measure $\mathbb P$ of Definition \ref{def:P} can also be defined using commutative stochastic calculus and the usual Girsanov theorem. Let $q\in \{1,\dotsc,k\}$ and the orthonormal basis $\mathbf{u}$ of $\mathbb C^k$ be the defining parameters of $\mathbb P$. Let $(\sigma_t)_{t\in\mathbb R_+}$ be the solution to the stochastic differential equation (SDE) $\sigma_0=\rho\in \mathcal D$ and 
\begin{equation}\label{eq:SDE_QTraj_linear}
\begin{split}
\d \sigma_t=& \mathcal L^*(\sigma_{t-})\d t\\
            & + \sum_{i=1}^q (L_{u_i}\sigma_{t-}+\sigma_{t-}L_{u_i})\d W_{i}(t)\\
            & + \sum_{j=q+1}^k \left(L_{u_j}\sigma_{t-}L_{u_j}^*-\sigma_{t-}\right)[\d N_j(t)-\d t],
\end{split}
\end{equation}
with 
$$\mathcal L^*:X\mapsto -i[H,X]+\tr_{\mathbb C^k}[\mathbf{L}X\mathbf{L}^*] -\frac12(\mathbf{L}^*.\mathbf{L}X+X\mathbf{L}^*.\mathbf{L}),$$
so that $\tr\circ\mathcal L^*=0$. The processes $W_1,\dotsc,W_q$ are independent Brownian motions and $N_{q+1},\dotsc,N_{k}$ are independent Poisson process of unit intensities, independent of the Brownian motions $W_1,\dotsc,W_q$. Denote $\mathbb Q$ the probability measure on $(\mathrm{D},\mathcal F)$ of $t\mapsto(W_1(t),\dotsc,W_q(t),N_{k+1}(t),\dotsc,N_k(t))$ and $\mathbb Q_t$ its restriction to $\mathcal F_t$. The process $t\mapsto \sigma_t$ is positive semi-definite valued (see \cite{Bouten2007,Barchielli1995}) and allows us to recover $\mathbb P$ as a Girsanov transform of $\mathbb Q$.
\begin{proposition}
Fix $q\in\{1,\dotsc,k\}$ and let $\mathbf{u}$ be an orthonormal basis of $\mathbb C^k$. Let $\rho\in \mathcal D$ and $\sigma$ be the solution of \eqref{eq:SDE_QTraj_linear}. Let $Z:t\mapsto \tr[\sigma_t]$. Then $Z$ is a positive Doléans-Dade exponential martingale and for any $t\in \mathbb R_+$,
$$\d \mathbb P_t=Z_t \d\mathbb Q_t.$$
\end{proposition}
\begin{proof}
See \cite{Bouten2007,Barchielli1995}.
\end{proof}
The process $\varrho:t\mapsto \sigma_t/\tr[\sigma_t]$ lives in the set of density matrices $\mathcal D$. It can be shown (see \cite{Bouten2007,Barchielli1995}) that with respect to $\mathbb P$, $\varrho$ is a quantum trajectory in the sense that it is the unique solution to
\begin{equation}\label{eq:SDE_QTraj}
\begin{split}
\d\varrho_t=&\mathcal L^*(\varrho_{t-})\d t\\
        &+ \sum_{i=1}^q \left(L_{u_i}\varrho_{t-} + \varrho_{t-}L_{u_i} - \tr[(L_{u_i}+L_{u_i}^*)\rho_{t-}]\varrho_{t-}\right)\d B_{u_i}(t)\\
        &+ \sum_{j=q+1}^k \left(\frac{L_{u_j}\varrho_{t-}L_{u_j}^*}{\tr[L_{u_j}^*L_{u_j}\varrho_{t-}]} -\varrho_{t-}\right)[\d N_{u_j}(t) - \tr[L_{u_j}^*L_{u_j}\varrho_{t-}]\d t].
\end{split}
\end{equation}
Here $B_{u_1},\dotsc,B_{u_q}$ are $q$ independent Brownian motions and $N_{q+1},\dots,N_{k}$ are $k-q$ inhomogeneous Poissons processes with stochastic intensities $\tr[L_{u_{q+1}}^*L_{u_{q+1}}\rho_{t-}]\d t,\dotsc,\tr[L_{u_k}^*L_{u_k}\rho_{t-}]\d t$ respectively, such that for each $j\in\{q+1,\dotsc,k\}$, $t\mapsto N_{u_j}(t)-\int_0^t \tr[L_{u_j}^*L_{u_j}\rho_{s-}]\d s$ is a martingale.
Moreover, following \cite{Bouten2007}, we deduce that with respect to $\mathbb P$, for any $i\in \{1,\dotsc q\}$
$$X_i\sim t\mapsto B_{u_i}(t) - \int_0^t \tr[(L_{u_i}+L_{u_i}^*)\varrho_{s-}]\d s$$
and for any $j\in\{q+1,\dotsc,k\}$,
$$Y_j\sim N_{j}.$$
This shows the relationship between the measurement signals $\mathbf X$ and $\mathbf Y$ and the underlying quantum system state $\varrho$. The state $\varrho$ determines the drift of the measurement signal $\mathbf X$ and the jump intensities of the measurement signal $\mathbf Y$. In particular, since $\mathbb E[\varrho_t]=\e^{t\mathcal L^*}\rho$, 
\begin{equation}\label{eq:average_signal}
\mathbb E[X_i]=\tr\left[(L_{u_i}+L_{u_i}^*)\int_0^t\e^{s\mathcal L^*}\rho\,\d s\right]\quad \quad{and}\quad\mathbb E[Y_j]=\tr\left[L_{u_j}^*L_{u_j}\int_0^t\e^{s\mathcal L^*}\rho\,\d s\right].
\end{equation}

The operator $\mathcal L^*$ in particular is the generator of a semigroup of trace preserving completely positive maps. In the next section we detail the properties of such semigroups and their duals.

\subsection{Quantum Markov semigroups}\label{subsec:qms}
In this section we give the definition of a quantum Markov semi-group, discuss its relationship to non-commutative Girsanov transforms and quantum trajectories and present some properties of symmetric quantum Markov semi-groups.

\subsubsection{Quantum Markov semi-groups}\label{sec:qms}

An open quantum system is said to undergo Markovian dynamics if its interaction with its environment is memoryless. In this case its dynamics is modeled by a quantum Markov semi-group (or quantum dynamical semi-group).

\begin{definition}[QMS]\label{def_qds}
A quantum Markov semi-group (QMS) $t\mapsto e^{t\mathcal L}$ is a uniformly continuous semi-group of completely positive\footnote{A map $\Phi:\mathcal B(\cH)\to\mathcal B(\cH)$ is completely positive if and only if $\Phi\otimes \id_{M_n(\mathbb C)}$ is a positive map for any natural  number $n$.} maps from $\mathcal B(\cH)$ to itself that preserves the identity: $e^{t\mathcal L}(\id_{\cH})=\id_{\cH}$. 

\end{definition}
The generator $\mathcal{L}$ of a QMS can always be written in Lindblad form (\cite{GL76b,GKS76}).
\begin{theorem}\label{thm_Lindblad}
If $\mathcal{L}$ is the generator of a $\operatorname{QMS}$, then there exist $H\in \mathcal B_{\operatorname{sa}}(\cH)$, $k\in \mathbb N$ and $\mathbf L:\cH\to\cH\otimes \mathbb C^k$, such that
\begin{equation}\label{eq_Lindblad_op}
\mathcal{L}:X\mapsto i[H,X]+\mathbf {L}^*(X\otimes \id_{\mathbb C^k})\mathbf {L}-\frac{1}{2}(\mathbf{L}^*.\mathbf{L}\,X+X\,\mathbf{L}^*.\mathbf{L}).
\end{equation}
Conversely, if $\mathcal{L}$ is as above, then it generates a $\operatorname{QMS}$. 
\end{theorem}
\noindent
The generator $\mathcal{L}$ is called a Lindbladian. The operator $\mathcal L_D:X\mapsto \mathcal L(X)-i[H,X]$ is sometimes called the dissipator. The dual $\mathcal L^*$ of $\mathcal L$ with respect to the Hilbert-Schmidt inner product is the generator of a semi-group of completely positive maps preserving the trace. In particular, $e^{t\mathcal L^*}\mathcal D(\cH)\subset \mathcal D(\cH)$. By abuse of notation we also call $t\mapsto e^{t\mathcal L^*}$ a QMS. 

Note that Theorem~\ref{thm_Lindblad} shows that the operators $\mathcal L^*$ in Eqs.~\eqref{eq:SDE_QTraj_linear} and~\eqref{eq:SDE_QTraj} are generators of a QMS. Conversely, for any QMS, we can define a process $\varrho$, which is a solution of a SDE like~\eqref{eq:SDE_QTraj}, such that $\mathbb E[\varrho_t]=\e^{t\mathcal L^*}\rho$. Such a process $\varrho$ is then called an unraveling of the QMS. Unravelings are used as numerical tools to study QMS (\cite{dalibard_wave-function_1992,gisin1992quantum}). 

More importantly for us, any QMS can be unitarily dilated (\cite[Theorem 7.3]{hudson1984stochastic}):
\begin{theorem}[Dilation]
Let $(s,t)\mapsto U_{t,s}$ be the two parameter unitary group solution of Eq.~\eqref{eq:def_U} with the same $H$ and $\mathbf{L}$ as in the definition of the generator $\mathcal L$. Then for any $X\in \mathcal B(\mathcal H)$, 
$$\e^{t\mathcal L}(X)=\tr_\Gamma[U_{t,0}^*(X\otimes\id_\Gamma) U_{t,0}(\id_{\cH}\otimes \Omega )].$$
The unitary group $(s,t)\mapsto U_{t,s}$ is called a dilation of the QMS.
\end{theorem}

Note that dilations and unravelings are not unique. The next proposition, translated from \cite[Proposition 7.4]{Wolf2011}, characterizes those $H$ and $\mathbf{L}$ which define the same $\operatorname{QMS}$.
\begin{proposition}\label{prop:up_to_unitary}
The operator couples $(H,\mathbf L)\in \mathcal B_{\operatorname{sa}}(\cH)\times\mathcal B(\ch;\cH\otimes \mathbb C^k)$ and $(H',\mathbf{L}')\in \mathcal B_{\operatorname{sa}}(\cH)\times\mathcal B(\ch;\cH\otimes \mathbb C^k)$ define the same $\operatorname{QMS}$ if and only if there exists a $k\times k$ unitary matrix $U$, a vector $\mathbf{c}\in \mathbb C^k$ and a real constant $E\in \mathbb R$ such that
$$\mathbf{L}'=(\id_\cH\otimes U)\mathbf{L}+\id_\cH\otimes \mathbf{c}$$
and
$$H'=H-\Im(\mathbf{L}^*(\id_\cH\otimes \mathbf{c}))+E.$$
\end{proposition}

\subsubsection{Dirichlet form}\label{sec:Dirichlet form}
A QMS generically has a unique stationary state \cite{Hanson2020}, i.e., there exists a unique $\sigma\in \mathcal D(\cH)$ such that $\mathcal L^*(\sigma)=0$. If this stationary state $\sigma$ is in addition of full rank, the corresponding QMS is called \emph{primitive}. In that case, one can show \cite{BCG13} that every initial state $\rho$ converges to $\sigma$ when evolved with respect to the QMS. From now on we will assume $t\mapsto e^{t\mathcal L}$ to be primitive and denote by $\sigma$ its invariant state with full-rank.

% As for classical Markov semi-groups, our study is based on the investigation of the Dirichlet form associated to the generator $\mathcal L$. As we did not introduce any notion of symmetry yet, we work with the symmetrized Dirichlet form. As in the commutative setting, we first equip $\mathcal B(\cH)$ with a Hilbert space structure.
Let us first equip $\mathcal B(\cH)$ with a Hilbert space structure.
\begin{definition}[KMS inner product]
The KMS inner product on $\mathcal B(\cH)$ is defined for $X,Y\in \mathcal B(\cH)$ by
$$\langle X,Y\rangle_{\operatorname{KMS}}=\tr[\sigma^{\frac12}X^*\sigma^{\frac12}Y].$$
Equipped with this inner product, $\mathcal B(\cH)$ is a Hilbert space. 
\end{definition}

As we most often use the KMS inner product, when no confusion is possible, we may sometimes write $\langle .,.\rangle_\sigma\equiv \langle .,.\rangle_{\operatorname{KMS}}$ to emphasize the state $\sigma$ with respect to which the inner product is defined. In what follows, we also refer to the norm associated to $\langle.,.\rangle_{\operatorname{KMS}}$ as $\|.\|_{L^2(\sigma)}$ and refer to the corresponding non-commutative weighted $L^2$ space as $L^2(\sigma)$. It will also be convenient to define the following operator:
\begin{align}
    \Gamma_\sigma:\cB(\cH)\to\cB(\cH)\,,\qquad  X\mapsto \sigma^{\frac{1}{2}}X\sigma^{\frac{1}{2}}\,. 
    \end{align}
For a map $\Phi:\mathcal B(\cH)\to\mathcal B(\cH)$ we denote by $\Phi^{\operatorname{KMS}}$ its dual with respect to the KMS inner product.

While we primarily work with KMS inner products, we will sometimes also use the GNS inner product.
\begin{definition}[GNS inner product]
The GNS inner product on $\mathcal B(\cH)$ is defined for $X,Y\in \mathcal B(\cH)$ by
$$\langle X,Y\rangle_{\operatorname{GNS}}=\tr[\sigma X^*Y].$$
Equipped with this inner product, $\mathcal B(\cH)$ is a Hilbert space.
\end{definition}

A last possible inner product we would like to mention here is the BKM inner product. While we do not use it, it is relevant to the study of semi-groups that are gradient flows of the entropy (\cite{carlen2017gradient}).
\begin{definition}[BKM inner product]
The BKM inner product on $\mathcal B(\cH)$ is defined for $X,Y\in \mathcal B(\cH)$ by
$$\langle X,Y\rangle_{\operatorname{BKM}}=\int_0^1\tr[\sigma^{1-t} X^*\sigma^t Y]\,\d t.$$
Equipped with this inner product, $\mathcal B(\cH)$ is a Hilbert space.
\end{definition}

We can now define the Dirichlet form associated to the generator $\mathcal L$.
\begin{definition}[Dirichlet form]
The (symmetrized) Dirichlet form of the QMS generator $\mathcal L$ is defined for any $X\in \mathcal B(\cH)$ by
$$\mathcal E_{\mathcal L}(X)=-\frac12\big(\langle X,\mathcal L(X)\rangle_{\operatorname{KMS}}+\langle \mathcal L(X),X\rangle_{\operatorname{KMS}}\big ).$$
\end{definition}
Note that if $\mathcal L$ is symmetric with respect to the KMS inner product, we recover the usual (non-symmetrized) Dirichlet form. In the next section we further investigate these symmetric generators of QMS.

In the rest of the article, to any map $T:\mathcal B(\mathcal H)\to\mathcal B(\mathcal H)$, $T^{\operatorname{KMS}}$ denotes its dual with respect to the KMS inner product.

\subsubsection{Quantum detailed balance}
For a QMS, the notion of symmetry is most often referred to through its physical interpretation of detailed balance. There exist different definitions of quantum detailed balance (QDB) depending on the inner product used on $\mathcal B(\cH)$. We focus on the GNS, BKM and KMS notions of QDB.
\begin{definition}\label{def:QDB}
The QMS, or equivalently its generator, is said to verify KMS, GNS or BKM QDB if it is symmetric with respect to the respective inner product.
\end{definition}

The GNS notion of QDB is the most restrictive in the sense that 
$$\mbox{GNS QDB}\implies \mbox{KMS and BKM QDB}.$$
A proof of this implication can be found in \cite[Theorem 2.9]{carlen2017gradient}. Counterexamples to KMS QDB implying GNS QDB can be found in \cite[Appendix B]{carlen2017gradient}. Similarly, counterexamples to BKM QDB implying GNS QDB can be constructed. In \cite{BCJPP}, an example of a BKM and KMS symmetric map that is not GNS symmetric is provided. One can also prove that the BKM and KMS QDB notions are not comparable, meaning that there exist maps that are one and not the other in both cases (see Appendix~\ref{app:BKMKMS}).

The notions of detailed balance can be extended by requiring symmetry only up to a unitary or anti-unitary mapping~\cite{BCJPP}. Using this generalization it was proven in the same reference that KMS QDB is equivalent to the vanishing of some notion of entropy production. However, we do not consider this generalization here as it would dramatically obscure our discussion.

We mentioned in Proposition~\ref{prop:up_to_unitary} that a generator $\mathcal L$ can be defined using different operators $H$ and $\mathbf{L}$. Following \cite{fagnola2010generators} and \cite[Theorem 4.4]{amorim2021complete}, KMS QDB singles out a subclass of these operators. From now on, we define the modular operator
\begin{align}
    &\Delta\equiv \Delta_\sigma:\cB(\cH)\to\cB(\cH)\,, \qquad X\mapsto \sigma X \sigma^{-1}\,.
\end{align}
\begin{theorem}\label{thm:KMS Kraus decompo}
The generator $\mathcal L$ verifies $\operatorname{KMS}$ $\operatorname{QDB}$ if and only if there exists $H\in \mathcal B_{\operatorname{sa}}(\cH)$ and $\mathbf L:\cH\to\cH\otimes \mathbb C^k$ such that Eq.~\eqref{eq_Lindblad_op} holds and
$$(\Delta^{\frac12}\circ\operatorname{adj}\otimes \id_{\mathbb C^k}) \mathbf{L}=\mathbf{L},$$
where $\operatorname{adj}:\mathcal B(\cH)\to\mathcal B(\cH);\ X\mapsto X^*$ and
$$H=\frac{i}{2}\int_0^\infty e^{-t\sigma^{\frac12}}[\mathbf{L}^*.\mathbf{L},\sigma^{\frac12}]e^{-t\sigma^{\frac12}}\d t=\frac{i}{2}\tanh\circ\log(\Delta^{\frac14})(\mathbf{L}^*.\mathbf{L}).$$
\end{theorem}
Assuming GNS QDB imposes a finer structure \cite{fagnola2007generators}.
\begin{theorem}\label{thm:GNS Kraus decompo}
Let $d=\dim\cH$. The generator $\mathcal L$ verifies $\operatorname{GNS}$ $\operatorname{QDB}$ if and only if there exists $\mathbf L:\cH\to\cH\otimes \mathbb C^{d^2}$ such that Eq.~\eqref{eq_Lindblad_op} holds with $H=0$ and 
$$(\Delta\otimes \id_{\mathbb C^{d^2}}) \mathbf{L}=(\id_\cH\otimes D)\mathbf{L}$$
for some diagonal $d^2\times d^2$ matrix $D$ unitarily equivalent to the modular operator $\Delta$.
\end{theorem}
Note that the operators $\mathbf{L}$ of Theorems~\ref{thm:KMS Kraus decompo} and~\ref{thm:GNS Kraus decompo} are generally distinct.
In the next section we will use the Dirichlet form to derive concentration inequalities for the probability measures of Definition~\ref{def:P}. Assuming KMS QDB, we then show that these inequalities are optimal.

\section{Bounding the tail probability of indirect measurements}\label{sec:tail}

% In this section we show how the quality of (certain) indirect quantum measurements in quantum systems, whose evolution is described by a primitive quantum Markov semigroup, can be related to a perturbed version of this semigroup. We will first focus on the case of a measured system consisting of one bosonic mode, and an estimator effectively describing a Brownian motion (Section~\ref{subsubsec:BM}). Subsequently, the generalization to the case of a measured system consisting of several bosonic modes (Section~\ref{subsubsec:BMm}), as well as an analogy for Poisson processes, (Section~\ref{subsubsec:PP}) are presented.

Consider a system $S$ with associated Hilbert space $\cH_S$, which is coupled to an environment $E$ described by the Fock space $\Gamma$ such that the evolution of the combined system $S\vee E$ is described by the unitary evolution $U_t$ defined through the quantum stochastic differential equation (cf.\ Eq.~\eqref{eq:def_U})
\begin{equation}\label{eq:diffeqU}
\d U_{t}U_t^*= \left(-i H -\frac12 \mathbf L^*.\mathbf L\right)\d t + \mathbf{L}.\d\mathbf{A}_{\mathbf e}^*(t) - \d \mathbf{A}_{\mathbf{e}}(t).\mathbf{L}^* ,\quad U_{0}=\id\ .
\end{equation}
This combined evolution corresponds to -- given that the environment is in the Fock space vacuum state $\Omega$ -- an evolution of the reduced system $S$ described by a quantum Markov semigroup $t\mapsto \e^{t\mathcal{L}}$ (cf.\ Section~\ref{sec:qms}). We will assume here and in the following that $t\mapsto \e^{t\cL}$ is primitive, i.e., has a unique full-rank stationary state denoted by $\sigma$.

The goal is to indirectly measure observables of the system $S$. In the following, we will consider observables of the form 
\begin{align}\label{indirectlymeasured}
    O^B(u):=\sum_{j=1}^k u_j\,(L_j+L_j^*)=L_u+L_u^*\qquad \text{ and }\qquad O^P(u):=L_u^*L_u
\end{align} 
for some $u\in \mathbb{R}^k$ with $\|u\|=1$,  where $L_j$ denotes the $j$-th component of $\mathbf L$, $j\in\{1,\dots,k\}$, and $L_u$ is defined as in \Cref{eq:Lu}. For the indirect measurement of the observables, here the system which is physically measured is a subsystem (i.e., a subset of the modes) of the environment $E$. Alternatively, one could introduce a separate measurement system $M$ associated to additional bosonic modes. To perform an indirect measurement of observables on $S$, the modes of $M$ would then have to be coupled to the system $S\vee E$ such that the resulting combined evolution of the system $S\vee E\vee M$ is of the form~\eqref{eq:diffeqU}, but with the vectors $\mathbf{A}_{\mathbf{e}}(t)$, $\mathbf{A}_{\mathbf{e}}^*(t)$ enlarged to also incorporate the modes of the measurement system, which are then coupled to the observables of interest. 

Let $1\le q\le \ell$ be natural numbers, and let $\mathbf{u}:=\{u_1,\cdots,u_\ell\}$ be a set of orthogonal normalized vectors in $\mathbb{R}^k$. We consider the following vector of observables to be (indirectly) measured:
\begin{equation*}
\mathbf{O}(\mathbf{u})=\left(O^B(u_1),\dots,O^B(u_q),O^P(u_{q+1}),\dots, O^P(u_\ell)\right)\, .
\end{equation*}
We recall that the corresponding processes can be simultaneously measured (cf.~\Cref{sec:noncommnoises}). In order to distinguish the Brownian and Poisson parts, we introduce the notations $\mathbf{u}^B:=(u_1,\dots, u_q)$ and $\mathbf{u}^P:=(u_{q+1},\dots, u_\ell)$. The corresponding estimator $\mathbf{E}$ is then defined as 
\begin{align}\label{generalestimator}
\mathbf{E}_{t,\mathbf{u}}=\frac{1}{t}\,U_t^*(\mathbf{A}_{\mathbf{u}^B}(t)+\mathbf{A}_{\mathbf{u}^B}^*(t),{\Lambda}_{\mathbf{u}^P}(t))\,U_t\equiv \frac{1}{t}\,(\mathbf{X}(t),\mathbf{Y}(t))\,,
\end{align}
where $\mathbf{X}$ and $\mathbf{Y}$ are the operator valued functions defined in \Cref{prop:nc_girsanov}.  We denote by ${m}_{\mathbf{u}}=(m_{\mathbf{u}^B}^B,m_{\mathbf{u}^P}^P)\in\mathbb{R}^\ell$ the vector with entries 
\begin{align*}
\left({m}_{\mathbf{u}}\right)_j\equiv {m}_{u_j}\coloneqq \left\{ 
\begin{aligned}
&\tr(\sigma\, O^B(u_j))\ ,&\quad j\in\{1,\dots,q\}\;\;\\
&\tr(\sigma\, O^P(u_j))\ ,&\quad j\in\{q+1,\dots,\ell\}\,.
\end{aligned}
\right. 
\end{align*}
% \begin{align}
    % &X_j(t)=\int_0^t\,U_s^*\,O^B(u_j)\,U_s\,ds \,+B_{u_j}(t)\,,\quad j\le q\\
    % &Y_j(t)=U_t^*\,\Lambda_{u_j}(t)\,U_t\,,\qquad \qquad \qquad  \qquad \quad j\ge q+1\,.
% \end{align}

% Note that in both cases there are restrictions concerning the choice of the observables to be measured indirectly (see Section~\ref{?}).  \lh{this could e.g. be clarified in the introduction?}\textcolor{blue}{Yes, when we motivate by a physical model.}

\subsection{The main upper bound}

Given $r\in \mathbb{R}^\ell_+$, we are interested in upper bounding the probability 
\begin{align}\label{tailproba}
  \mathbb{P}(\cap_i\{E_{t,u_i}-m_{u_i}\ge r_i \})\,.
\end{align}
To this end, let us define $\Phi_{u_i}(X):=L_{u_i}^*X+X L_{u_i}$, $i\le q$, and $\Psi_{u_j}(X):=L_{u_j}^*XL_{u_j}$, $j\ge q+1$. Furthermore, let
\begin{align*}
&f_{u_i}^B(X):=\frac{1}{2}\,\langle X,\,(\Phi_{u_i}+\Phi_{u_i}^{\operatorname{KMS}})(X)\rangle_\sigma\,,~~~~~~~ i\le q\\
&f_{u_j}^P(X):=\frac{1}{2}\,\langle X,\,(\Psi_{u_j}+\Psi_{u_j}^{\operatorname{KMS}})(X)\rangle_\sigma\,,~~~~~~~ j\ge q+1\,,
\end{align*}
and $({f}_{\mathbf{u}^B}^B)_i(X):=f_{u_i}^B(X)$, $({f}_{\mathbf{u}^P}^P)_j(X):=f_{u_j}^P(X)$. 

We recall that the relative entropy between (possibly non-normalized) mass functions $\mathbf{p}$ and $\mathbf{q}$ is defined as
\begin{align*}
D(\mathbf{p}\|\mathbf{q}):=\sum_{l}\,p_l\,\ln\Big(\frac{p_l}{q_l}\Big)-p_l+q_l\ge0\,.
\end{align*}
Note that $D(\mathbf{p}\|\mathbf{q})=0$ if and only if $\mathbf{p}=\mathbf{q}$, and that it may be infinite if $q_l=0$ and $p_l>0$ for some $l$. 

The next theorem constitutes the main result of this section. Its proof is postponed to \Cref{sec:proofmain}.
\begin{theorem}\label{thm:mainbound}
	Let $t\mapsto \e^{t\cL}$ be a primitive quantum Markov semigroup on the finite dimensional matrix algebra $M_n(\mathbb{C})$ with invariant state $\sigma$. Then, for all $t\ge 0$, all initial states $\rho\in\cD(\mathbb{C}^n)$, any ${r}\in\RR_+^{\ell}$ and any family $\mathbf{u}=\{u_1,\dots,u_\ell\}$ of orthogonal, normalized vectors in $\mathbb{C}^k$:
\begin{align}\label{mainbound}
&\frac{\mathbb{P}(\cap_i\{E_{t,u_i}-m_{u_i}\ge r_i \})}{\|\Gamma_\sigma^{-1}(\rho)\|_{L^2(\sigma)}}\\
&\qquad\le       \exp\Big(-t\underset{\|X\|_{L^2(\sigma)}= 1}{\inf}\,\big\{\cE(X)+\frac{1}{2}\,\|{r}^B+{m}_{\mathbf{u}^B}^B-{f}_{\mathbf{u}^B}^B(X)\|^2+D({r}^P+{m}_{\mathbf{u}^P}^P\|{f}_{\mathbf{u}^P}^P(X))\big\}\Big)\,.\nonumber
\end{align}
\end{theorem}
In \Cref{sec:largedeviations}, we show that this concentration bound is optimal in the sense that for KMS symmetric semigroups and appropriately chosen Kraus operators, for large times we obtain a large deviation principle for $\mathbb P$ with rate function given by
\begin{align}\label{eq:def_sym_rate_function}
I_{\mathbf{u}}({s}):=\underset{\|X\|_{L^2(\sigma)}= 1}{\inf}\,\big\{\cE(X)+\tfrac{1}{2}\,\|{s}^B-{f}^B_{\mathbf{u}^B}(X)\|^2+D({s}^P\|{f}^P_{\mathbf{u}^P}(X))\big\}.
\end{align}
%  \begin{remark}
% \textcolor{red}{add comparison to classical bound}
% \end{remark}
\subsection{Reversible semigroups: optimal concentration bound and large deviation principle}\label{sec:largedeviations}

Following the proofs of \cite{JPW}, for any ${\lambda}\in \mathbb R^\ell$,
$$e({\lambda}):=\lim_{t\to\infty} \frac1t\ln\mathbb E(\exp(t\,{\lambda}.\mathbf{E}_{t,\mathbf{u}}))$$
exists and ${\lambda}\mapsto e({\lambda})$ is differentiable everywhere. It follows then from the Gärtner-Ellis theorem that the estimator $\mathbf{E}$ verifies a large deviation principle in the sense that for any Borel set $B\subset \mathbb R^\ell$,
$$-\inf_{{s}\in \operatorname{int}(B)}J({s})\leq \liminf_{t\to\infty}\frac1t\ln \mathbb P(\mathbf{E}_{t,\mathbf{u}}\in B)\leq\limsup_{t\to\infty}\frac1t\ln \mathbb P(\mathbf{E}_{t,\mathbf{u}}\in B)\leq-\inf_{{s}\in \operatorname{cl}(B)}J({s})$$
with $\operatorname{int}(B)$ the interior of $B$, $\operatorname{cl}(B)$ the closure of $B$ and
\begin{align}\label{eq:def_rate_function}
J:{s}\mapsto\sup_{{\lambda}\in \mathbb R^\ell}{\lambda}.{s}-e({\lambda})
\end{align}
taking values in $[0,+\infty]$ and being lower semi-continuous and convex as the supremum of affine continuous functions. The next theorem provides sufficient conditions ensuring $J=I$.
\begin{theorem}\label{LDP}
Assume $\mathcal{L}$ verifies $\operatorname{KMS}$ $\operatorname{QDB}$. Let $H\in \mathcal B(\cH)$ and $\mathbf L:\cH\to\cH\otimes \mathbb C^k$ be as in Theorem~\ref{thm:KMS Kraus decompo} and let $(s,t)\mapsto U_{t,s}$ be the two parameter unitary group solution of Eq.~\eqref{eq:def_U} with these two operators. Then for $U_T=U_{T,0}$, any orthonormal family $\mathbf u=\{u_1,\dots,u_\ell\}$ of $\mathbb{R}^k$ and any $q\in\{1,\dotsc,\ell\}$, the measure $\mathbb P$ of Definition~\ref{def:P} is such that for any ${s}\in \mathbb R^\ell$
$$J(s)=I(s)$$
with $J:\mathbb R^k\to[0,+\infty]$ the large deviation principle rate function defined by Eq.~\eqref{eq:def_rate_function} and $I:\mathbb R^k\to[0,+\infty]$ the concentration bound defined by Eq.~\eqref{eq:def_sym_rate_function}.
\end{theorem}
\begin{proof}
Since $(\Delta^{\frac12}\circ\operatorname{adj}\otimes\id_{\mathbb C^k})\mathbf{L}=\mathbf{L}$ and for any $i\in \{1,\dotsc,\ell\}$ $u_i$ is a real unit vector, $\sigma^{\frac12}L_{u_i}^*\sigma^{-\frac12}=L_{u_i}$. It then follows that $\Phi_i:X\mapsto L_{u_i}^*X+XL_{u_i}$ and $\Psi_i:X\mapsto L_{u_i}^*XL_{u_i}$ are both KMS symmetric. Therefore, for any $\lambda\in \mathbb R^\ell$, the perturbed generator $\mathcal L_{ \lambda,\mathbf{u}}$ defined in Eq.~\eqref{eq:perturbedgenerator} in the proof of Theorem~\ref{thm:mainbound} is KMS symmetric. Hence the spectrum of $\mathcal L_{\lambda,\mathbf{u}}$ is real and
$$\max\{\Re(x):x\in \spec\mathcal L_{ \lambda,\mathbf{u}}\}=\max\{x:x\in \spec\mathcal L_{ \lambda,\mathbf{u}}\}=\sup_{\|X\|_{L^2(\sigma)}=1}\langle X,\mathcal L_{ \lambda,\mathbf{u}}(X)\rangle_{\operatorname{KMS}}=-\inf_{\|X\|_{L^2(\sigma)}=1}\mathcal E_{\lambda}(X).$$
From a direct extension of \cite{JPW}, $e(\lambda)=\max\{\Re(x):x\in \spec\mathcal L_{ \lambda,\mathbf{u}}\}$. Hence,
$$e(\lambda)=-\inf_{\|X\|_{L^2(\sigma)}=1}\mathcal E_{\lambda}(X).$$
Since $J$ is defined as the Legendre transform of the left hand side and, following the proof of Theorem~\ref{thm:mainbound}, particularly Eq.~\eqref{eq:usefulconcentra}, $I$ is the Legendre transform of the right hand side, we deduce $J=I$.
\end{proof}

The last theorem proves that the bound in Theorem~\ref{thm:mainbound} is optimal. Indeed, if ${r}\in \operatorname{int}(\{ x \in \mathbb R_+^\ell : I({m}_{\mathbf{u}}+ x)<\infty\})$ the convexity of $I$ implies it is continuous in a neighborhood of ${m}_{\mathbf{u}}+{r}$, and therefore  minimal in ${m}_{\mathbf{u}}+{r}$ on both the interior and closure of $\{{s}\in \mathbb R^\ell: s_i-m_{u_i}\ge r_i,\, \forall i\in \{1,\dotsc,\ell\} \}$. It follows that the inequalities in the large deviation principle are saturated and
$$\lim_{t\to\infty}\frac1t\ln\mathbb P(\cap_i\{E_{t,u_i}-m_i\ge r_i \})=-I({m}_{\mathbf{u}}+ r).$$

\section{Proof of \Cref{thm:mainbound}}\label{sec:proofmain}

This section is dedicated to the proof of \Cref{thm:mainbound}. In \Cref{sec:fksg}, we show that there exists an upper bound on the tail probability \eqref{tailproba} that depends on a certain perturbed semigroup. As pointed out in the subsequent \Cref{contractiontodirichlet}, we can in particular bound the tail probaility in terms of the $L^2(\sigma)\to L^2(\sigma)$ contraction of this perturbed semigroup. Moreover, we show how to further bound the latter contraction in terms of the Dirichlet form of the generator $\cL$ corresponding to the original semigroup.

\subsection{The quantum perturbed semigroup}\label{sec:fksg}
In this section, we show an upper bound on the tail probability in terms of a perturbed semigroup:
\begin{proposition}\label{prop:multivariatebound}
	Let $t\mapsto \e^{t\cL}$ be a primitive quantum Markov semigroup on the finite dimensional matrix algebra $M_n(\mathbb{C})$ with invariant state $\sigma$. Furthermore, let  $1\le q\le \ell$, $\lambda\equiv (\lambda^B,\lambda^P)\in\RR_+^{\ell}$, with $\lambda^B=(\lambda_1,...,\lambda_q)$ and $\lambda^P=(\lambda_{q+1},...,\lambda_\ell)$, and let $\mathbf{u}\equiv(\mathbf{u}^B,\mathbf{u}^P)$, with $\mathbf{u}^B =\{u_1,\dots,u_q\}$ and $\mathbf{u}^P =\{u_{q+1},\dots,u_\ell\}$, be a family of orthogonal, normalized vectors in $\mathbb{C}^k$. Then, for all $t\ge 0$, all initial states $\rho\in\cD(\mathbb{C}^n)$, any $r\in \mathbb{R}_+^\ell,\lambda$, and $\mathbf{u}$, we have
\begin{align}\label{mainbound_lambda}
\mathbb{P}(\cap_i\{E_{t,u_i}-m_{u_i}\ge r_i \})\le \tr(\rho\,\e^{t\cL_{\lambda,\mathbf{u}}}(\operatorname{id}))\,\e^{-t\,\lambda. (m_{\mathbf{u}}+r)}\,,
\end{align}
where 
\begin{align}\mathcal{L}_{{\lambda},\mathbf{u}}(X)\coloneqq \mathcal{L}(X)+ {\lambda}^B.\left(\mathbf{L}_{\mathbf{u}^B}^*X+X\mathbf{L}_{\mathbf{u}^B}+\frac{{\lambda}^B}{2}X\right)+\sum_{j=q+1}^\ell\,(\e^{\lambda_j}-1)L_{u_j}^*XL_{u_j}\,.\label{eq:perturbedgenerator}
\end{align}
\end{proposition}

Before we prove \Cref{prop:multivariatebound} in its full generality in \Cref{subsubsec:BMm}, we first consider two different simplified setups in Sections~\ref{subsubsec:BM} and \ref{subsubsec:PP}. These setups treat the special cases of the estimator $\mathbf{E}_{t,\mathbf{u}}$ corresponding to a single-mode Brownian motion ($\mathbf{E}_{t,\mathbf{u}}\equiv E_{t,u_1^B}$), and a single-mode Poisson process ($\mathbf{E}_{t,\mathbf{u}}\equiv E_{t,u_{1}^P}$) respectively. These simplified examples will be instructive for the more general multivariate case (with both kinds of processes combined).

\subsubsection{Brownian motion}\label{subsubsec:BM}
% In this section, we chose $u\in\mathbb{R}^k$ normalized, i.e.~$\|u\|=1$, and let $O^B(u)$ as in \Cref{indirectlymeasured} be the observable to be (indirectly) measured.

% We recall that the expectation of the observable $O^B(u)$ in the stationary state $\sigma$ is defined by 
% \begin{align}\label{eq:defm}
% m^B_u\coloneqq \tr\left(\sigma \,O^B(u)\right)\ ,
% \end{align}
% and define 
% \begin{align}\label{eq:estimator}
% E_{t,u}^B\coloneqq U_t^*\left(A_{u}(t)+A_{u}^*(t)\right)U_t\ .
% \end{align}
% We call $E_{t,u}^B$ the estimator corresponding to $O^B(u)$. By construction $t\mapsto E_{t,u}^B$ is the process $X_0$ as defined in Poposition~\ref{prop:nc_girsanov} where we have set $u_0=u$. %Note that the family $A_{u}(t)+A_{u}^*(t)$ consists of observables on the environmental system $E$ alone, and together with the vacuum $\Omega$ on $\Gamma$ is stochastically equivalent to a standard Brownian motion (cf.\ Theorem~\ref{thm:mult_by_B} in Section~\ref{sec:noncommnoises}). 
In this section we consider the special case of the estimator 
\begin{align*}
    \mathbf{E}_{t,\mathbf{u}}\equiv E_{t,u_1^B}=\frac{1}{t}U_t^*(A_{u_1^B}(t)+A_{u_1^B}^*(t))U_t=\frac{1}{t}X_1(t)\,,
\end{align*}
which corresponds to a single-mode Brownian motion (cf.\ \Cref{thm:mult_by_B}). We will denote $u:=u_1^B$ and $X(t):=X_1(t)$ for simplicity throughout this section.

\begin{proposition}\label{prop_smbm}
Let $t\mapsto \e^{t\cL}$ be a primitive quantum Markov semigroup on $M_n(\mathbb{C})$ with invariant state $\sigma$. Then, for all $t\geq 0$, all initial states $\rho\in\cD(\mathbb{C}^n)$, any $r,\lambda\in\mathbb{R}_+$ and $u\in S^{k-1}(\mathbb{R})$, we have
\begin{align*}
    \mathbb{P}\left(X(t)/t-m_u^B>r\right)\le \tr\big(\rho\,\e^{t\cL_{\lambda,u}^B}(\id)\big)\,\e^{-\lambda t (m_u^B+r)}\,,
\end{align*} 
where 
\begin{align*}
    \cL_{\lambda,u}^B(X):=\mathcal{L}(X)+ \lambda(L_u^*X+XL_u)+\frac{\lambda^2}{2}X\,.
\end{align*}
\end{proposition}

\begin{proof}
Let $\phi_{t,u}^B(\lambda):=\mathbb{E}[\exp(i\lambda X(t))]$ be the characteristic function of $X(t)$. By \Cref{prop:nc_girsanov} we may write  
\begin{align*}
   \phi_{t,u}^B(\lambda)=\tr\left[\rho\otimes \Omega\ \,\e^{i\lambda X(t)}\right]\,,
\end{align*}
where the expression on the right side of the above equation is well-defined since $\e^{i\lambda X(t)}$ is unitary. Furthermore, by the definition of $X(t)$, we have that
\begin{align*}
   \phi_{t,u}^B(\lambda)&=\tr\left[\rho\otimes\Omega\ \e^{i\frac{\lambda}{2}U_t^*(A_{u}(t)+A_{u}^*(t))U_t}\,\id\,\e^{i\frac{\lambda}{2}U_t^*(A_{u}(t)+A_{u}^*(t))U_t}\right]\\
&=  \tr\left[\rho\otimes\Omega\ U_t^*\e^{i\frac{\lambda}{2}(A_{u}(t)+A_{u}^*(t))}\,\id\,\e^{i\frac{\lambda}{2}(A_{u}(t)+A_{u}^*(t))}U_t\right]\\
&=  \tr\left[\rho\,{\Phi_{t,u}^B}^{(i\lambda)}(\id)\right]\,,
\end{align*}
where the family of operators $t\mapsto {\Phi_{t,u}^B}^{(i\lambda)}$ on $M_n(\mathbb{C})$ is given by
\begin{align}\label{def:perturbedsemigroupsmbm}
{\Phi_{t,u}^B}^{(i\lambda)}(X):=\tr_\Gamma\left[\Omega\ {V_{t,u}^B}^*(-\lambda)XV_{t,u}^B(\lambda)\right]\,,
\end{align}
with $V_{t,u}^B(\lambda):= \e^{i\frac{\lambda}{2}(A_{u}(t)+A_{u}^*(t))}U_t$ unitary. This is a strongly continuous semigroup with respect to $t$, with generator (cf.\ Appendix \ref{app:subsec:compgenbm})
\begin{align}\label{def:gensmbm}
\mathcal{L}_{i\lambda,u}^B(X)\coloneqq \mathcal{L}(X)+ i\lambda\, (L_u^*X+X L_u)-\frac{\lambda^2}{2}X\,.
\end{align}
Moreover, the expression of its generator ensures that for any $t$, ${\Phi_{t,u}^B}^{(i\lambda)}$ is an entire analytic function in $\lambda$. This in particular also implies that the characteristic function $\phi_{t,u}^B(\lambda)$ is entire analytic in $\lambda$, whence it can be written as a Fourier integral on the whole complex plane (cf.\ Theorem 7.1.1 of \cite{lukacs1970characteristic}), and therefore the Laplace transform of $X(t)$, $\varphi_{t,u}^B(\lambda):=\mathbb E[\exp(\lambda X(t))]$, is well-defined. In particular, the latter can be expressed in terms of the analytic continuation $ {\Phi_{t,u}^B}^{(\lambda)}$ of ${\Phi_{t,u}^B}^{(i\lambda)}$:
\begin{align*}
    \varphi_{t,u}^B(\lambda)= \phi_{t,u}^B(-i\lambda)=\tr\left[\rho\,{\Phi_{t,u}^B}^{(\lambda)}(\id)\right]\,.
\end{align*}
Hence, using that for any $\lambda>0$, $x\in \mathbb R\mapsto \exp(\lambda x)\in \mathbb R_+$ is strictly increasing, together with Markov's inequality, we get for any $r\geq 0$ and $\lambda>0$
\begin{align*}
\mathbb{P}\left(X(t)/t-m_u^B>r\right)&=  \mathbb{P}\left(\e^{\lambda X(t)}>\e^{\lambda t\left(m_u^B+r\right)}\right)\nonumber\\ 
&\leq  \varphi_{t,u}^B(\lambda)\,\e^{-\lambda t\left(m_u^B+r\right)}\\
&=\tr\left[\rho\,{\Phi_{t,u}^B}^{(\lambda)}(\id)\right]\,\e^{-\lambda t\left(m_u^B+r\right)} 
\end{align*}
as claimed.
\end{proof}

\begin{remark}
\renewcommand{\labelenumi}{(\alph{enumi})}
\begin{enumerate}
    \item \label{rem:absvalprob} Note that we have only established a bound on the right tail of $X(t)/t$. However, a derivation analogous to the one in the above proof also yields a bound for the left tail: For any $r>0$ and $\lambda>0$,
$$\mathbb P(X(t)/t-m_u^B<-r)\leq \varphi_{t,u}^B(-\lambda)\,\e^{\lambda t(m_u^B- r)}.$$
Using this, together with \Cref{prop_smbm}, a double-sided bound can be established. %Furthermore, note that such bounds are also known as Chernoff's bounds. 
\item Note that we may write
\begin{align*}
{\Phi_{t,u}^B}^{(i\lambda)}(X)
%&=\tr_\Gamma\left(\Omega\ U_t^*X\e^{i\lambda(A_{u}(t)+A_{u}^*(t))}U_t\right)\nonumber\\
&=\tr_\Gamma\left[\Omega\ U_t^*XU_t\e^{i\lambda U_t^*(A_{u}(t)+A_{u}^*(t))U_t}\right]\,,\nonumber
\end{align*}
where we used that the operators $A_{u}(t)+A_{u}^*(t)$ act on $\Gamma$ only, as well as the unitarity of $U_t$. Thus the semigroup $t\mapsto {\Phi_{t,u}^B}^{(i\lambda)}$ emerges from the quantum Markov semigroup $t\mapsto \e^{t\cL}$ via the introduction of the perturbation (by \Cref{prop:nc_girsanov} and linearity in $u$)
\begin{align}
\e^{i\lambda X(t)}=\exp\left(i\lambda \left(\int_0^tU_s(O^B(u))U_s^*\,\d s+A_{u}(t)+A_{u}^*(t)\right)\right)\,,\nonumber
\end{align}
i.e., up to Brownian motion, a perturbation corresponding to the time-averaged evolution of $O^B(u)\equiv L_u+L_u^*$.
\end{enumerate}
\end{remark}

\subsubsection{Poisson process}\label{subsubsec:PP}
As in the previous section, here we consider a special case of the estimator $\mathbf{E}_{t,\mathbf{u}}$:
\begin{align*}
    \mathbf{E}_{t,\mathbf{u}}\equiv E_{t,u_1^P}=\frac{1}{t}U_t^*\Lambda_{u_1^P}(t)U_t=\frac{1}{t}Y_1(t)\,.
\end{align*}
 We will denote $u:=u_1^P$ and $Y(t):=Y_1(t)$ for simplicity throughout this section.

\begin{proposition}\label{prop_smpp}
Let $t\mapsto \e^{t\cL}$ be a primitive quantum Markov semigroup on $M_n(\mathbb{C})$ with invariant state $\sigma$. Then, for all $t\geq 0$, all initial states $\rho\in\cD(\mathbb{C}^n)$, any $r,\lambda\in\mathbb{R}_+$ and $u\in S^{k-1}(\mathbb{R})$, we have
\begin{align*}
    \mathbb{P}\left(Y(t)/t-m_u^P>r\right)\le \tr\big[\rho\,\e^{t\cL_{\lambda,u}^P}(\id)\big]\,\e^{-\lambda t (m_u^P+r)}\,,
\end{align*} 
where 
\begin{align}\label{def:gensmpp}
    \cL_{\lambda,u}^P(X):=\mathcal{L}(X)+ (\e^{\lambda}-1)L_u^*XL_u\,.
\end{align}
\end{proposition}

\begin{proof}
The proof of \Cref{prop_smpp} follows the same lines as the proof of \Cref{prop_smbm}. First of all, we relate the characteristic function $\phi_{t,u}^P(\lambda):=\mathbb{E}[\exp(i\lambda Y(t))]$ of $Y(t)$ to the strongly continuous semigroup $t\mapsto {\Phi_{t,u}^P}^{(i\lambda)}$ on $M_n(\mathbb{C})$,
\begin{align}\label{def:perturbedsemigroupsmpp}
{\Phi_{t,u}^P}^{(i\lambda)}(X):=\tr_\Gamma\left[\Omega\ {V_{t,u}^P}^*(-\lambda)XV_{t,u}^P(\lambda)\right]
\end{align}
with $V_{t,u}^P(\lambda):= \e^{i\frac{\lambda}{2}\Lambda_{u}(t)}U_t$ unitary and generator (cf.\ Appendix~\ref{app:subsec:compgenpoisson})
\begin{align*}
\mathcal{L}_{i\lambda,u}^P(X)\coloneqq \mathcal{L}(X)+ (\e^{i\lambda}-1)L_u^*X L_u\,,
\end{align*}
using similar arguments as in the case of Brownian motion (cf.\ \Cref{subsubsec:BM}). From this, again following the same argumentation as in the previous section treating single-mode Brownian motion, we get the well-defined Laplace transform of $Y(t)$ of the form 
\begin{align*}
    \varphi_{t,u}^P(\lambda):=\mathbb E(\exp(\lambda Y(t)))= \phi_{t,u}^P(-i\lambda)=\tr\left[\rho\,{\Phi_{t,u}^P}^{(\lambda)}(\id)\right]\,, 
\end{align*}
and the properties of the exponential function, together with Markov's inequality, subsequently yield the claimed bound on the tail probability.
\end{proof}
\begin{remark}
\renewcommand{\labelenumi}{(\alph{enumi})}
\begin{enumerate}
    \item \label{rem:absvalprobpp}
Note that, as in the case of Brownian motion, the bound on the left tail is established analogously, yielding for any $r>0$ and $\lambda>0$:
$$\mathbb P(Y(t)/t-m_u^P<-r)\leq \varphi_{t,u}^P(-\lambda)\,\e^{\lambda t(m_u^P- r)}.$$
Using this, together with \Cref{prop_smpp}, a double-sided bound can be established. 
\item Note that, similarly to the Brownian motion setting, the semigroup $t\mapsto {\Phi_{t,u}^P}^{(i\lambda)}$ can be seen as a perturbed version of the quantum Markov semigroup $t\mapsto \e^{t\cL}$. More precisely, we may write
\begin{align*}
{\Phi_{t,u}^P}^{(i\lambda)}(X)
&=\tr_\Gamma\left[\Omega\ U_t^*XU_t\e^{i\lambda U_t^*\Lambda_{u}(t)U_t}\right]\,,\nonumber
\end{align*}
where we used that the operators $\Lambda_{u}(t)$ act on $\Gamma$ only, as well as the unitarity of $U_t$. Here the perturbation corresponds to the time-averaged evolution of the observable $O^P(u)\equiv L_u^*L_u$ (up to an additive stochastic part, cf.\ \Cref{prop:nc_girsanov}, and noting that $Y(t)$ is linear in $uu^*$).
\end{enumerate}
\end{remark}

\subsubsection{Proof of \Cref{prop:multivariatebound}}\label{subsubsec:BMm}
Let us now consider the more general setup of \Cref{prop:multivariatebound} and the general estimator $\mathbf{E}_{t,\mathbf{u}}\equiv (\bX(t),\bY(t))/t$. Recall that $1\le q\le \ell$, $\lambda\equiv (\lambda^B,\lambda^P)\in\RR_+^{\ell}$, with $\lambda^B=(\lambda_1,...,\lambda_q)$ and $\lambda^P=(\lambda_{q+1},...,\lambda_\ell)$, and $\mathbf{u}\equiv(\mathbf{u}^B,\mathbf{u}^P)$, with $\mathbf{u}^B =\{u_1,\dots,u_q\}$ and $\mathbf{u}^P =\{u_{q+1},\dots,u_\ell\}$, is a family of orthogonal, normalized vectors in $\mathbb{C}^k$.

To prove \Cref{prop:multivariatebound}, we proceed as in both the special cases of Sections~\ref{subsubsec:BM} and \ref{subsubsec:PP}: We first relate the characteristic function $\phi_{t,\bu}(\lambda):=\mathbb{E}[\exp(i \lambda. (\bX(t),\bY(t)))]$ of the vector $(\bX(t),\bY(t))$ to a perturbed semigroup $t\mapsto {\Phi_{t,\bu}}^{(i\lambda)}(X)$, which then yields a well-defined expression of the Laplace transform of the same vector, whence the properties of the exponential function, together with Markov's inequality, yield the claimed bound. 

In particular, the perturbed semigroup $t\mapsto {\Phi_{t,\bu}}^{(i\lambda)}(X)$ on $M_n(\mathbb{C})$ is given by 
\begin{align}\label{def:perturbedsemigroupmm}
{\Phi_{t,\bu}}^{(i\lambda)}(X):=\tr_\Gamma\left[\Omega\ {V_{t,\bu}}^*(-\lambda)XV_{t,\bu}(\lambda)\right]\,,
\end{align}
with $V_{t,\bu}(\lambda):= \e^{\frac{i}{2}\lambda.(\bA_{\bu^B}(t)+\bA_{\bu^B}^*(t),\Lambda_{\bu^P}(t))}U_t$ unitary. The generator of this strongly continuous semigroup with respect to $t$ is (cf.\ Appendix \ref{app:subsec:compgenmm})
\begin{align}\label{def:genmm}
\mathcal{L}_{i{\lambda},\mathbf{u}}(X)\coloneqq &\mathcal{L}(X)+ i{\lambda}^B.\left(\mathbf{L}_{\mathbf{u}^B}^*X+X\mathbf{L}_{\mathbf{u}^B}+\frac{i{\lambda}^B}{2}X\right)+\sum_{j=q+1}^\ell\,(\e^{i\lambda_j}-1)L_{u_j}^*XL_{u_j} \,.
\end{align}
The identity
\begin{align*}
   \phi_{t,\bu}(\lambda)=  \tr\left[\rho\,{\Phi_{t,\bu}}^{(i\lambda)}(\id)\right]\,,
\end{align*}
the resulting form of the Laplace transform of $(\bX(t),\bY(t))$,
\begin{align*}
    \varphi_{t,\bu}(\lambda):= \mathbb E[\exp(\lambda. (\bX(t),\bY(t)))] =\phi_{t,\bu}(-i\lambda)=\tr\left[\rho\,{\Phi_{t,\bu}}^{(\lambda)}(\id)\right]\,,
\end{align*}
and also the final resulting bound on the tail probability are then established analogously to the case of single-mode Brownian motion (cf.\ \Cref{subsubsec:BM}) and single-mode Poisson process (cf.\ \Cref{subsubsec:PP}).

\begin{remark}
Note that, as in the special cases of Sections~\ref{subsubsec:BM} and \ref{subsubsec:PP}, we can also get a double-sided bound using the same techniques as in the proof of \Cref{prop:multivariatebound}.
\end{remark}

\subsection{The upper bound}\label{contractiontodirichlet}
% {\color{blue}Assume $\mathcal L$ is KMS symmetric. Let $\mathcal V$ be the smallest real vector space containing the operators $\{L(k)\}_{k=1}^{n-1}$. Then, following Theorem~\ref{thm:KMS_decomposition}, for any $O\in \mathcal V$, there exists $p_O\in \mathbb R_+$, $\mathbf{x}\in S^{n-2}(\mathbb R)$ such that $O=p_O\sum_{k=1}^{n-1}x_k \Gamma_k$ with $\{\Gamma_k\}$ the set of operators defining $\mathcal L$ defined in Theorem~\ref{thm:KMS_decomposition}.

% For any $O\in\mathcal V$ and $\lambda\in\mathbb R$, let 
% $$\mathcal L_{O,\lambda}:X\mapsto \mathcal L(X)+\lambda(O^*X+XO)+\frac{\lambda^2}{2}p_O^2X.$$
% The Laplace transform of $X_O(t)=\tfrac1t\int_0^t\tr(\rho_s(O^*+O))\d s+\frac{p_O}{t}W_t$ is given by
% $$\mathbb E(\exp(t\lambda X_O(t)))=\tr(\sigma\exp(t\mathcal L_{O,\lambda})(\one)).$$

% Using Theorem \ref{thm:KMS_decomposition}, the KMS symmetry of $\mathcal L$ implies the KMS symmetry of $\mathcal L_{O,\lambda}$.}

Here we upper bound the probability that the estimator $\mathbf{E}_{t,\mathbf{u}}$ defined in \Cref{generalestimator} is away from the mean vector ${m}_{\mathbf{u}}$ in terms of the Dirichlet form of $\cL$. First, we rewrite more explicitly the perturbed generator:
\begin{align}
\mathcal{L}_{{\lambda},\mathbf{u}}(X)&=\mathcal{L}(X)+\sum_{j=1}^q\,\lambda_j\,(L_{u_j}^{*}X+XL_{u_j}+\frac{1}{2}\lambda_jX)+\sum_{j=q+1}^\ell\,\big(\e^{\lambda_j}-1  \big)\,L_{u_j}^*XL_{u_j}\label{primaryexpremap}\\
 &\quad=\sum_{i=1}^kL_i^* X L_i-\frac{1}{2}\{L_i^*L_i,X\}+\sum_{j=1}^q\,\lambda_j\,(L_{u_j}^{*}X+XL_{u_j}+\frac{1}{2}\lambda_jX)+\sum_{j=q+1}^\ell\big(\e^{\lambda_j}-1  \big)\,L_{u_j}^*XL_{u_j}\,.\nonumber
 \end{align}
 Next, we let $\lambda^B(\mathbf{u}^B)\in \mathbb{R}^k$ be the vector of components $\lambda^B(\mathbf{u}^B)_i=\sum_{j=1}^q\lambda_{j} (u_j)_i$. Since for each $j$, $u_j$ is assumed to be normalized, we further have that $\|\lambda^B(\mathbf{u}^B)\|_2=\|\lambda^B\|_2$. Then,
 \begin{align}
     \sum_{i=1}^k&L_i^* X L_i+\sum_{j=1}^q\,\lambda_j\,(L_{u_j}^{*}X+XL_{u_j}+\frac{1}{2}\lambda_jX) 
     =\sum_{i=1}^k(L_i+\lambda^B(\mathbf{u}^B)_i)^*X(L_i+\lambda^B(\mathbf{u}^B)_i)-\frac{1}{2}\|\lambda^B\|_2^2\,X\,.
 \end{align}
Therefore 
\begin{align}\label{eq:LCP}
 \cL_{\lambda,\mathbf{u}}(X)  = \Psi_{{\lambda},\mathbf{u}}(X)-\frac{1}{2}\|\lambda^B\|_2^2X-\frac{1}{2}\sum_{i=1}^k\{L_i^*L_i,X\}\,,
\end{align}
for some completely positive map $\Psi_{{\lambda},\mathbf{u}}$ defined as
\begin{align}
\Psi_{{\lambda},\mathbf{u}}(X):=\sum_{j=q+1}^\ell (e^{\lambda_j}-1)L_{u_j}^*X L_{u_j}+\sum_{i=1}^k\,(L_i+\lambda^B(\mathbf{u}^B)_i)^*X(L_i+\lambda^B(\mathbf{u}^B)_i)\,.
\end{align}
% , and where the sums over the Brownian and Poissonian deformations are disjoint. 
Then, for $\rho:=\Gamma_\sigma(X)\equiv \sigma^{\frac{1}{2}}X\sigma^{\frac{1}{2}}$, we have
\begin{align}\label{eq:CSstep}
\varphi_{t,\mathbf{u}}({\lambda})=\tr\big[\rho\,\Phi^{{\lambda}}_{t,\mathbf{u}}(\id)\big]\le \|X\|_{L^2(\sigma)}\,\|{\Phi}_{t,\mathbf{u}}^{{\lambda}}:\,L^2(\sigma)\to L^2(\sigma)\|\,.
\end{align}
 By the Lumer-Philips theorem we can upper bound the last operator norm as follows:
\begin{align}\label{eq:boundFKsemi}
\|{\Phi}_{t,\mathbf{u}}^{{\lambda}}:\,L^2(\sigma)\to L^2(\sigma)\|\le \exp\big({-t\underset{\|X\|_{L^2(\sigma)}=1}{\inf}  \cE_{{\lambda}}(X)}\big)\,,
\end{align}
where $\cE_{{\lambda},\mathbf{u}}(X):=-\frac{1}{2}(\langle X,\cL_{{\lambda},\mathbf{u}}(X)\rangle_\sigma+\langle \cL_{{\lambda},\mathbf{u}}(X),X\rangle_\sigma)$ denotes the (symmetrized) Dirichlet form of $\cL_{{\lambda},\mathbf{u}}$. Optimizing over ${\lambda}$, we end up with
\begin{align}\label{eq:usefulconcentra}
\mathbb{P}(\cap_i\{E_{t,u_i}-m_{u_i}\ge r_i \})\le\,\|\Gamma_\sigma^{-1}(\rho)\|_{L^2(\sigma)}\, \exp\Big(-t\,\underset{{\lambda}}{\sup}\,\,\,\underset{\|X\|_{L^2(\sigma)}= 1}{\inf}\,  \big\{\cE_{{\lambda},\mathbf{u}}(X)+{\lambda}. ({r}+{m}_{\mathbf{u}}) \big\} \Big)\,.
\end{align}
Following \Cref{eq:LCP}, $\cL_{{\lambda},\mathbf{u}}+\cL_{{\lambda},\mathbf{u}}^{\rm KMS}$ is the generator of a completely positive semi-group. Hence, by the Perron--Frobenius Theorem, the infinimum is a minimum that is reached for $X$ positive semi-definite. The infinimum over $X$ can thus be restricted to $X$ positive semi-definite.
Inspired by this observation, we define the function $g:\mathbb{R}_+^{\ell_B+\ell_P}\times\cD(\cH)\to\RR$ as
\begin{align*}
g({\lambda},\gamma):=\cE_{{\lambda},\mathbf{u}}(\Gamma_\sigma^{-\frac{1}{2}}(\gamma^{\frac{1}{2}}))+{\lambda}.({r}+{m}_{\mathbf{u}})\,\,,
\end{align*}
so that the optimization in the exponential on the right-hand side of \eqref{eq:usefulconcentra} turns out to be equivalent to $\sup_{{\lambda}}\inf_{\gamma\in\cD(\cH)}g({\lambda},\gamma)$. 
\begin{lemma}
The function $g$ is convex in the state $\gamma$ and concave in the parameter ${\lambda}$.
\end{lemma}
\begin{proof}
The concavity in ${\lambda}$ can be directly verified from \eqref{primaryexpremap}. To prove the convexity in $\gamma$, we consider the Dirichlet form 
\begin{align*}
\cE_{{\lambda},\mathbf{u}}(X)&=-\frac{1}{2}\langle X,\,(\cL_{{\lambda},\mathbf{u}}+\cL_{{\lambda},\mathbf{u}}^{\operatorname{KMS}})(X)\rangle_\sigma  \\
&= -\frac{1}{2}\langle  X,\,(\Psi_{{\lambda},\mathbf{u}}+\Psi_{{\lambda},\mathbf{u}}^{\operatorname{KMS}})(X)\rangle_\sigma +\frac{\|\lambda^B\|^2}{2}\,\|X\|_{L^2(\sigma)}^2-\frac{1}{2}\langle X,\,LX+XL'\rangle_\sigma\,,
\end{align*}
for some operators $L,L'$ depending on $\mathbf{u}$ and $\lambda$. Since the KMS-dual of a completely positive map is completely positive, the map $\Psi_{{\lambda},\mathbf{u}}+\Psi_{{\lambda},\mathbf{u}}^{\operatorname{KMS}}$ has a Kraus decomposition which we denote by $(\Psi_{{\lambda},\mathbf{u}}+\Psi_{{\lambda},\mathbf{u}}^{\operatorname{KMS}})(X):=\sum_{l}K_l^*XK_l$. Now, for $X=\Gamma_\sigma^{-\frac{1}{2}}(\sqrt{\gamma})$, $\|X\|_{L^2(\sigma)}=1$ and denoting $\widetilde{K}_l:=\sigma^{-\frac{1}{4}}K_l\sigma^{\frac{1}{4}}$, we have
\begin{align*}
\cE_{{\lambda},\mathbf{u}}(\Gamma_\sigma^{-\frac{1}{2}}(\sqrt{\gamma}))=-\frac{1}{2}\sum_l\,\tr\Big[  \sqrt{\gamma}\,\widetilde{K}_l^* \sqrt{\gamma} \widetilde{K}_l\Big]+\frac{1}{2}\|{\lambda}^B\|^2-\frac{1}{2}\,\tr\big[\Gamma_\sigma^{-\frac{1}{2}}(\sqrt{\gamma})\sigma^{\frac{1}{2}} \big[ L\Gamma_\sigma^{-\frac{1}{2}}(\sqrt{\gamma})+\Gamma_\sigma^{-\frac{1}{2}}(\sqrt{\gamma})L'\big]  \sigma^{\frac{1}{2}}\big]\,.
\end{align*}
By the Ando-Lieb concavity theorem (see Theorem 5.15 of \cite{Wolf2011}), the first sum over $l$ is convex in $\gamma$. Moreover, 
\begin{align*}
\tr\big[\Gamma_\sigma^{-\frac{1}{2}}(\sqrt{\gamma})\sigma^{\frac{1}{2}} \big[ L\Gamma_\sigma^{-\frac{1}{2}}(\sqrt{\gamma})+\Gamma_\sigma^{-\frac{1}{2}}(\sqrt{\gamma})L'\big]   \sigma^{\frac{1}{2}}\big]=\tr\big[ \gamma (\sigma^{\frac{1}{4}}L\sigma^{-\frac{1}{4}}+\sigma^{-\frac{1}{4}}L'\sigma^{\frac{1}{4}} )\big]
\end{align*}
is linear in $\gamma$. This ends the proof.
\end{proof}
We are now ready to prove \Cref{thm:mainbound}:
\begin{proof}[Proof of \Cref{thm:mainbound}]

Thanks to the previous Lemma, we can use Sion's minimax theorem in order to swap the minimization in $X$ and the suppremum in $\lambda$ in \eqref{eq:usefulconcentra}, so that
\begin{align*}
\mathbb{P}(\cap_i\{E_{t,u_i}-m_{u_i}\ge r_i \})\le\,\|\Gamma_\sigma^{-1}(\rho)\|_{L^2(\sigma)}\, \exp\Big(-t\,\underset{\|X\|_{L^2(\sigma)}= 1}{\inf}\, \underset{{\lambda}}{\sup} \big\{\cE_{{\lambda},\mathbf{u}}(X)+{\lambda}. ({r}+{m}_{\mathbf{u}}) \big\} \Big)\,.
\end{align*}
 A simple optimization over $\lambda$ then yields:
\begin{align*}
\sup_{{\lambda}}\{\cE_{{\lambda},\mathbf{u}}(X)+{\lambda}.({r}+{m}_{\mathbf{u}})\}=\cE(X)+\frac{1}{2}\,\|{r}^B+{m}_{\mathbf{u}^B}^B-{f}_{\mathbf{u}^B}^B(X)\|^2+D({r}^P+{m}_{\mathbf{u}^P}^P\|{f}_{\mathbf{u}^P}^P(X))\,,
\end{align*}
 and the result follows.
\end{proof}
% ----------------------------------------------

% Then, for any $X\in\dom(\cL_\lambda)$, $\|X\|_{L^2(\sigma)}=1$, \textcolor{red}{justify that we can take $\Delta_\sigma$ in infinite dimensions}, defining $M_k(X):=\{O^*(k)+\Delta_\sigma^{-\frac{1}{2}}(O(k))\}X+X\{O(k)+\Delta_\sigma^{\frac{1}{2}}(O^*(k))\}$ we have 
% \begin{align*}
% \cE^{O+O^*}_\lambda(X)+\lambda(r+m)=\frac{1}{2}\,\langle X,-(\cL+\cL^\sigma)(X)\rangle_\sigma-\frac{\lambda}{2} \sum_{k=1}^n\langle X,M_k(X)\rangle_\sigma-\frac{n\lambda^2}{2}+\lambda (r+m)
% \end{align*}
% The above expression is concave as a function of $\lambda$. To make it convex in $X$, we add and subtract $ \frac{\lambda}{2}\|\sum_k M_k:L^2(\sigma)\toL^2(\sigma)\|<\infty$ by assumed boundedness of the operators $O(k)$. Therefore, we can swap the minimization and maximization due to the minimax theorem. Moreover, defining 
% \begin{align}\label{eq:mu}
% \mu(X):= \sum_{k=1}^n \langle X,M_k(X)\rangle_\sigma,
% \end{align}
% we have for all $X\in\dom(\cL)$,
% \begin{align*}
% \sup_{\lambda >0}\big\{\cE^{O+O^*}_\lambda(X)+\lambda(r+m)\big\}=\frac{1}{2}\langle X,-(\cL+\cL^\sigma)(X)\rangle_\sigma +\frac{1}{2n}\,\big[  r+m-\frac{1}{2}\mu(X)\big]^2\,.
% \end{align*}

%

\section{Concentration via transportation and functional inequalities}\label{sec:concentration}

The goal of this section is to prove concentration for the tail probability in \Cref{thm:mainbound} by means of non-commutative functional and transportation cost inequalities. The main tool that we use is a lower bound on the Dirichlet form in terms of a quantum generalization of the Wasserstein distance. In \Cref{Wassersteindist}, we introduce a definition for the quantum Wasserstein distance which generalizes various quantities recently introduced in the community of quantum information theorists. \Cref{sec:functional} consists of a short review on functional and transportation cost inequalities, which we then use in \Cref{sec:concentration1} to derive our concentration bounds. In this section, we exclusively assume that our primitive QMS $t\mapsto \e^{t\mathcal{L}}$ over a finite dimensional Hilbert space is GNS symmetric.

% \subsection{Transportation-cost inequalities and concentration of  observables}

\subsection{Quantum transportation cost distances}\label{Wassersteindist}

Given a complete separable metric space $(\cX,d)$, let $c:\cX\times \cX\to [0,\infty]$ be a lower semicontinuous function such that $c(x,x)=0$ for all $x\in\cX$. The function $c$ is referred to as the \textit{cost function}. Given such a cost function, the \textit{transportation cost} $T_c$ is defined on the space of probability measures over $\cX$ by
\begin{align}\label{primal}
T_c(\mu,\nu):=\inf_{\pi\in\Omega(\mu,\nu)}\,\iint_{\cX^2}\,c(x,y)\,\pi(\d x,\d y)\,,
\end{align} 
where the infimum is taken over the set of couplings $\Omega(\mu,\nu)$ of $\mu$ and $\nu$. Whenever the cost function $c$ is equal to a power $d^p$ of the metric, $p\ge 1$, the transportation cost is usually denoted by $W_p$ and called the \textit{$L^p$ Wasserstein distance}. Transportation cost distances admit a dual representation, also known as the Kantorovich duality theorem \cite{villani2003topics}:
\begin{align}\label{dual}
T_c(\mu,\nu):=\sup_{(u,v)\in \Phi_c}\,\int\,u\d\nu-\int \,v\,\d\mu\,,
\end{align}
where $\Phi_c:=\{(u,v)\in \cB(\cX):\,u(x)-v(y)\le c(x,y) \,,\forall(x,y)\in\cX^2\}$ and $\cB(\cX)$ denotes the space of Borel-measurable, real bounded functions over $\cX$. This functional characterization is the one which Bobkov and Goetze used in their pioneering work \cite{bobkov1999exponential} on concentration inequalities. 
The dual expression was later extended as follows \cite{Gozlan2007,Guillin2009}:
\begin{align}\label{eqtransportcost}
\cT_\Phi(\nu,\mu)=\sup_{(u,v)\in\Phi}\int u\,\d\nu-\int v\,\d\mu\,,
\end{align}
where $\Phi\subset \cB(\cX)^2$ is a non-empty set such that (i) $u\le v$ for all $(u,v)\in\Phi$; and (ii) for all probability measures $\nu_1,\nu_2$, there exists $(u,v)\in\Phi$ such that $\int u\,\d\nu_1-\int v\,\d\nu_2\ge 0$. In the quantum setting, various extensions of transportation costs have been recently proposed \cite{junge2015noncommutative,rouze2019concentration,carlen2018non,de2019quantum,palma2020quantum}. As in the classical setting, these Wasserstein distances have in common that they can be written in terms of a supremum over test observables satisfying some linear constraints. Here, we adopt the approach of \cite{Gozlan2007,Guillin2009} and propose a unifying definition. For the sake of simplicity, we will restrict ourselves to finite dimensional systems:

\begin{definition}\label{def:TCmeas}
Let $\cH$ be a finite dimensional Hilbert space. Then, given a subset $\Phi$ of $\cB_{\sa}(\cH)^2$ such that 
\begin{itemize}
	\item[(i)] $X\le Y$ for all $(X,Y)\in\Phi$\,;
	\item[(ii)] for all $\omega_1,\omega_2\in \cD(\cH)$, there exists $(X,Y)\in\Phi$ such that $\tr[\omega_1 X]-\tr[\omega_2 Y]\ge 0$\,;
%	\item[(iii)] For all $(X,Y)\in \Phi$, $\sigma X=X\sigma$ and $\sigma Y=Y\sigma$ \textcolor{red}{maybe we can weaken this condition} 
	\end{itemize}
the \textit{quantum transportation cost distance} $\cT_{\Phi}:\cD(\cH)\times \cD(\cH)\to [0,\infty]$ is defined as
\begin{align*}
\cT_{\Phi}(\omega_1,\omega_2):=\sup_{(X,Y)\in \Phi}\,\tr[\omega_1 X]-\tr[\omega_2 Y]\,.
\end{align*}
		\end{definition}

\begin{example}[Trace distance] When $\Phi:=\{(T,T)|\, 0\le T\le \id \}$ is the set of quantum effects, $\cT_{\Phi}$ is the trace distance. 
	\end{example}
	
	\begin{example}[Wasserstein distance from Lindblad evolutions] Here we fix a set $\{\partial_j\equiv[L_j,.]\}_{j\in \mathcal{J}}=\{[L_j^*,.]\}_{j\in \mathcal{J}}$ of derivations compatible with a full-rank state $\sigma$, i.e., for which there exists $\{\omega_j\}_{j\in\mathcal{J}}$ such that for all $j\in\mathcal{J}$
		\begin{align*}
		\sigma \,L_j=\e^{-\omega_j}L_j\sigma\,,\quad L_j^*\sigma=\e^{\omega_j}\sigma L_j^*\,.
		\end{align*}
	Then the \textit{Wasserstein distance of order $1$} between two states $\rho,\omega\in\mathcal{D}(\mathcal{H})$ was defined in  \cite{rouze2019concentration} as
	\begin{align}\label{def:W1}
		W_{1,\cL}(\omega_1,\omega_2):=\sup_{\|X\|_{\operatorname{Lip}}\le 1} \,|\tr[\omega_1 X]-\tr[\omega_2 X]|\,,
		\end{align}
	where $$\|X\|_{\operatorname{Lip}}:=\Big(\sum_{j\in \mathcal{J}}  (\e^{-\omega_j/2}+\e^{\omega_j/2})\,\|\partial_jX\|_\infty^2 \Big)^{1/2}\,.$$ Here $\cL$ stands for the generator of the GNS-symmetric quantum Markov semigroup obtained from taking the operators $L_j$ to be its Lindblad operators.
		 This is nothing but the transportation cost $\cT_{\Phi}$ for $\Phi=\{(X,X)$, $\|X\|_{\operatorname{Lip}}\le 1\}$. A variant of this distance for symmetric generators, where the Lipschitz constant is based on the non-commutative gradient associated with the generator considered, can also be found in \cite{gao2018fisher}.
		 \end{example}
		\begin{example}[Quantum Ornstein distance]
	More recently, a new quantum Wasserstein distance on the $n$-fold tensor product $\cH^{\otimes n}$ was proposed in \cite{palma2020quantum} (see also \cite{de2021quantum}). It can be expressed in its Kantorovich dual form as 
	\begin{align*}
	W_{1}^{\operatorname{Orn}}(\rho,\sigma):=\sup_{\|X\|_{\operatorname{Lip}}^{\operatorname{Orn}}\le 1}\tr[X(\rho-\sigma)]\,,
	\end{align*}
	where 
	\begin{align*}
	\|X\|_{\operatorname{Lip}}^{\operatorname{Orn}}:=2\max_{i\in [n]}\min_{H^{(i)}\in\cB_{\operatorname{sa}}(\cH_{i^c})}\big\|H-\id_{i}\otimes H^{(i)}\big\|_\infty\,.
	\end{align*}
\end{example}

\subsection{Functional and transportation cost inequalities}\label{sec:functional}

Here we fix a primitive, GNS symmetric QMS $t\mapsto e^{t\cL}$ with invariant state $\sigma$. Following the standard classical notations of \cite{Guillin2009}, we denote the Dirichlet form of a normalized, positive semi-definite operator $X\in L^2(\sigma)$ as
 \begin{align*}
I_{\cL}(\rho) :=\cE(X)\,,
 \end{align*}
 where $\rho=(\sigma^{\frac{1}{4}}X\sigma^{\frac{1}{4}})^2\equiv (\Gamma_\sigma^{\frac{1}{2}}(X))^2$. This extends the definition of the Fisher-Donsker-Varadhan information $I(\nu|\mu):=\cE_L(\sqrt{f})$ with $\nu=f\mu$ as defined in \Cref{classDirich}. It is also the quantity that arises for example on the right-hand side of \Cref{mainbound}. This observation is at the core of the main result of the section, namely the derivation of concentration bounds for quantum trajectories based on the transportation cost-information inequalities which we define now (see \cite{Guillin2009} for classical analogues).
 \begin{definition}\label{defTI}
 	Let $t\mapsto \e^{t\cL}$ be a primitive quantum Markov semigroup with invariant state $\sigma$, $\cT_{\Phi}$ a quantum transportation cost distance and  $\alpha:[0,\infty)\to [0,\infty]$ a function that is left-continuous, increasing and such that $\alpha(0)=0$. Then the triple $(\cL,\cT_{\Phi},\alpha)$ is said to satisfy a \textit{quantum transportation cost-information inequality} if the following holds: for any $\rho\in \cD(\cH)$,
 	\begin{align}\tag{$\alpha({\operatorname{\cT_{\Phi})I}}$}\label{def:TI}
 	\alpha(\cT_{\Phi}(\rho,\sigma))\le  I_{\cL}(\rho)\,.
	\end{align}
 \end{definition}
Next, we provide examples of triples $(\cL,\cT_{\Phi},\alpha)$ satisfying a quantum transportation cost-information inequality. We do so by relating the latter to previously studied quantum functional and transportation cost inequalities. We recall that the \textit{quantum entropy functional} for any $X\ge 0$ is defined as
\begin{align*}
\operatorname{Ent}_{2,\sigma}(X):=\tr\Big[ \Gamma_\sigma^{\frac{1}{2}}(X)^2\,(\ln\Gamma_\sigma^{\frac{1}{2}}(X)^2-\ln\sigma)  \Big]-\|X\|_{L^2(\sigma)}^2\,\ln\|X\|_{L^2(\sigma)}^2\,.
\end{align*}
Whenever $\|X\|_{L^2(\sigma)}=1$, the quantum entropy functional is equal to the relative entropy 
\begin{align*}
    D(\rho\|\sigma):=\tr\big[\rho\,(\ln\rho-\ln\sigma)\big]\,,
\end{align*}
for $\rho:=\Gamma_\sigma^{\frac{1}{2}}(X)^2$.
 Next, the \textit{entropy production} of the semigroup is defined for any $\rho\in\cD(\cH)$ as 
 \begin{align*}
 \operatorname{EP}_{\mathcal{L}}(\rho):=-\tr\big[\mathcal{L}^*(\rho)(\ln\rho-\ln\sigma)\big]\,\,.
 \end{align*}
Finally, we recall that the variance of $X$ is defined as
\begin{align*}
    \operatorname{Var}_\sigma(X):=\|X-\tr[\sigma X]\|_{L^2(\sigma)}^2\,.
\end{align*}
	 Then the QMS is said to satisfy:
\begin{itemize}
\item[(i)] A \textit{logarithmic Sobolev inequality} if there exists a constant $\alpha_2>0$ such that, for all $X\ge 0$,
	\begin{align}\tag{LSI($\alpha_2$)}
	\alpha_2\,\operatorname{Ent}_{2,\sigma}(X)\le \mathcal{E}(X)\,.
	\end{align}
	We denote the best constant achieving this bound by $\alpha_2(\cL)$.
	\item[(ii)] A \textit{modified logarithmic Sobolev inequality} if there exists a constant $\alpha_1>0$ such that, for any $\rho\in\cD(\cH)$:
	\begin{align}\tag{MLSI($\alpha_1$)}
	4\alpha_1\,D(\rho\|\sigma)\le \operatorname{EP}_{\mathcal{L}}(\rho)\,.
	\end{align}
		We denote the best constant achieving this bound by $\alpha_1(\cL)$.
	\item[(iii)] A \textit{transportation cost inequality} if there exists $c>0$ such that, for all $\rho\in\cD(\cH)$:
	\begin{align}\tag{TC($c$)}
	W_{1,\cL}(\rho,\sigma)\le \sqrt{2c\,D(\rho\|\sigma)}\,.
	\end{align}
% 	in other words, \eqref{def:TC} holds for $\alpha(r)=\frac{r^2}{2c}$, $E_*(\rho)=\sigma$ for all $\rho$ and $\cT_\Phi:=W_{1,\cL}$. 
	\item[(iv)] A \textit{transportation cost-information inequality} if there exists $C>0$ such that, for all $\rho\in\cD(\cH)$:
	\begin{align}\tag{TI(C)}
	W_{1,\cL}(\rho,\sigma)\le \sqrt{2C\,\operatorname{I}_{\mathcal{L}}(\rho)}\,;
	\end{align}
		in other words, \eqref{def:TI} holds for $\alpha(r)=\frac{r^2}{2C}$, $E_*(\rho)=\sigma$ for all $\rho$ and $\cT_\Phi:=W_{1,\cL}$. 
% 		\item[(v)] A $1$-\textit{transportation-information inequality} if there exists $C>0$ such that, for all $\rho\in\cD(\cH)$:
% 	\begin{align}\tag{TI$_1$(C)}
% 	W_{1,\cL}(\rho,\sigma)\le \sqrt{2C\,\operatorname{I}_{\mathcal{L},1}(\rho)}\,;
% 	\end{align}
	\item[(v)] A \textit{Poincar\'{e} inequality} if there exists $\lambda>0$ such that, for all $X\in\cB(\cH)$:
	\begin{align}\tag{PI($\lambda$)}\label{eq:Poincare}
	\lambda\,\operatorname{Var}_\sigma(X)\le \mathcal{E}(X)\,.
	\end{align}
	The best constant achieving this inequality is the spectral gap of $\cL$, which we denote by $\lambda(\cL)$.
\end{itemize}
The following proposition regroups known results connecting the inequalities introduced above, and relating them to the notion of a quantum transportation cost-information inequality:
\begin{proposition}\label{linkfunctc}
Let $t\mapsto \e^{t\cL}$ be a finite dimensional, $\operatorname{GNS}$ symmetric quantum Markov semigroup. Then the following implications hold:
\begin{align}
\operatorname{LSI}(\alpha_1)\quad &\Rightarrow \quad \operatorname{MLSI}(2\alpha_1)\tag{i}\\
\operatorname{MLSI}(\alpha_1)\quad &\Rightarrow \quad \operatorname{TC}((4\alpha_1)^{-1})\tag{ii}\\
\operatorname{TC}(c)+\operatorname{LSI}(\alpha_2)\quad &\Rightarrow\quad \operatorname{TI}\Big(\frac{c}{\alpha_2}\Big)\tag{iii}\\
%  \operatorname{TC}(c)+\operatorname{MLSI}(\alpha_1)\quad &\Rightarrow\quad //
%  \operatorname{TI}_1\Big(\frac{c}{4\alpha_1}\Big)\tag{iv}\\
		\operatorname{MLSI}(\alpha_1)\quad& \Rightarrow \quad \operatorname{PI}(2\alpha_1)\,.\tag{iv}
	\end{align}
	\end{proposition}

\begin{proof}
The proof of (i) and (iv) can be found in \cite{[KT13],Carbone2014}. (ii) was derived in \cite{rouze2019concentration}. Finally (iii) and (iv) are direct consequences of the definitions of LSI and MLSI. 
	\end{proof}
\begin{corollary}\label{TIfromLSI}
Let $t\mapsto \e^{t\cL}$ be a finite dimensional, $\operatorname{GNS}$ symmetric quantum Markov semigroup, and assume that $\operatorname{LSI}(\alpha_2)$ holds. Then $\operatorname{TI}(8^{-1}\alpha_2^{-2})$ holds. 
\end{corollary}
\begin{remark}
In a paper to appear \cite{GaoRouzecurvature}, one of the authors proves the that the transportation-information inequality is satisfied under a certain condition of positivity of a non-commutative version of Ollivier's coarse Ricci curvature lower bound \cite{Ollivier2009}, hence generalizing a result by Fathi and Shu \cite{Fathi2018}. The latter is satisfied e.g., for a family of quantum Gibbs samplers.
\end{remark}

Next, we also prove the following transportation cost-information inequality based on the Poincar\'{e} inequality (see Theorem 3.1 in \cite{Guillin2009}):
\begin{proposition}\label{prop:poincartoTI} The Poincaré inequality $\operatorname{PI(\lambda)}$ implies the following transportation cost-information inequality: for any $\rho\in\cD(\cH)$,
	\begin{align}\label{eq:trdirich}
	\|\rho-\sigma\|_1^2\,\le\, \frac{4}{\lambda}\,I_{\cL}(\rho)\,.
	\end{align}
\end{proposition}

\begin{proof}
	We first prove the following rudimentary inequality: for all $X\ge 0$,
	\begin{align}\label{eq:tracevar}
	\|\Gamma_\sigma(X)-\sigma\|_1^2\le 4\operatorname{Var}_\sigma(I_{2,1}(X)) \,,
	\end{align}
where $I_{2,1}(X):=\Gamma_\sigma^{-\frac{1}{2}}\big[(\Gamma_{\sigma}(X))^{\frac{1}{2}}\big]$. Indeed, 
	\begin{align*}
	\|\rho-\sigma\|_1&=\frac{1}{2}\big\|\big\{ \sqrt{\rho}-\sqrt{\sigma} ,\sqrt{\rho}+\sqrt{\sigma}\big\}\big\|_1=\frac{1}{2}\,\|  \{ \sigma^{\frac{1}{4}}(\sigma^{-\frac{1}{4}}\sqrt{\rho}\,\sigma^{-\frac{1}{4}}-\id)\,\sigma^{\frac{1}{4}},\,\sqrt{\rho}+\sqrt{\sigma}\}\|_1\\
	&\le \frac{1}{2}\,\big(\|\sigma^{\frac{1}{4}}(\sigma^{-\frac{1}{4}}\sqrt{\rho}\,\sigma^{-\frac{1}{4}}-\id)\,\sigma^{\frac{1}{4}}\,(\sqrt{\rho}+\sqrt{\sigma})\|_1+\|(\sqrt{\rho}+\sqrt{\sigma})\,\sigma^{\frac{1}{4}}(\sigma^{-\frac{1}{4}}\,\sqrt{\rho}\,\sigma^{-\frac{1}{4}}-\id)\,\sigma^{\frac{1}{4}}\|_1\big)\\
	&\le \|\sigma^{\frac{1}{4}}(\sigma^{-\frac{1}{4}}\,\sqrt{\rho}\,\sigma^{-\frac{1}{4}}-\id)\,\sigma^{\frac{1}{4}}\|_2\,\|
	\sqrt{\rho}+\sqrt{\sigma}\|_2\\
	&\le 2\,\|\sigma^{-\frac{1}{4}}\,\sqrt{\rho}\,\sigma^{-\frac{1}{4}}-\id\|_{L^2(\sigma)}\,,
	\end{align*}
	and \Cref{eq:tracevar} follows after taking the square and choosing $\rho=\Gamma_\sigma(X)$. Equation (\ref{eq:trdirich}) follows from a direct application of Poincar\'{e}'s inequality and the definition of $I_\cL(\rho)$.
\end{proof}

\subsection{Concentration of trajectories}\label{sec:concentration1}

Concentration of measure is the phenomenon according to which almost all the points of a set are
close to a subset of positive measure. More precisely, let $(\cX,d)$ be a metric space, and $\mu$ a probability measure on the Borel sets $\cB(\cX)$. Then, given a set $A \in \cB(\cX)$ such that $\mu(A)\ge 1/2$, the complement $(A^r)^c$ of its r-enlargements $A^r=\{x\in\cX;d(x,A)\le r\}$ should rapidly decay with $r$. This is typically the case when a transportation cost inequality of the following form is satisfied for the measure $\mu$: there exists a positive function $\alpha$ such that for any other measure $\nu \ll \mu$, 
\begin{align*}
	\alpha(W_1(\nu,\mu))\le D(\nu\|\mu)\,.
	\end{align*}
 This connection was first proved by Marton in \cite{[M86]} by a beautiful geometric argument. A more analytical argument based on bounds on the Laplace transform of $\mu$ was established later by Bobkov and G\"{o}tze who further proved the equivalence between Gaussian concentration of $\mu$ and the corresponding transportation cost inequality for the Wasserstein distance and $\alpha(r)=r^2$ \cite{bobkov1999exponential}. 	
In \cite{rouze2019concentration}, the authors extended the approach of Bobkov and G\"{o}tze to the quantum setting and proved that the transportation cost inequality for the Wasserstein distance defined in (\ref{def:W1}) and $\alpha(r)=r^2$ implies Gaussian concentration of Lipschitz quantum observables. Similar proofs can also be found in \cite{gao2018fisher,palma2020quantum,de2021quantum}.

In \cite{Guillin2009}, the functional analytical approach of Bobkov and G\"{o}tze was extended to Markov processes. There, concentration for observables evolving along the stochastic process was proven to be equivalent to the existence of a transporation-information inequality. The main theorem of this section is inspired by Theorem 2.2 of \cite{Guillin2009}:
\begin{theorem}\label{theo:ti}
	Let $t\mapsto \e^{t\cL}$ be a primitive, $\operatorname{GNS}$-symmetric finite dimensional quantum Markov semigroup on $\cB(\cH)$ with invariant state $\sigma$, and let $\Phi\subset \cB_{\sa}(\cH)^2$ be as in \Cref{def:TCmeas}. Assume further that $\alpha({\operatorname{\cT_{\Phi}})\operatorname{I}}$ holds with $\alpha(r):=\frac{r^2}{2C}$. Then, for any initial state $\rho\in\cD(\cH)$, $t>0$, $u\in \mathbb{C}^{k}$, and any indirectly measured observable $O^B(u):=L_{u}+L^*_u$ such that the observable $\widetilde{O}^B(u):=\Delta_\sigma^{\frac{1}{4}}(L^*_u)+\Delta_\sigma^{-\frac{1}{4}}(L_u)$ satisfies $(\sqrt{C}\widetilde{O}^B(u),\sqrt{C}\widetilde{O}^B(u))\in\Phi$,
\begin{align}\label{mainboundTIp_gauss}
\mathbb{P}\Big(\frac{1}{t}\int_0^t\tr[O^B(u)\rho_s]\,ds+\|u\|\,\frac{W_t}{t} >\tr[\sigma O^B(u)]+r\Big)\le     \|\Gamma_\sigma^{-1}(\rho)\|_{\mathbb{L}_{2}(\sigma)}\, \exp\Big(-\frac{t\,r^2}{4}\Big)\,.
\end{align}			
	\end{theorem}
\begin{proof}
Without loss of generality, we assume that $\|u\|=1$ and $\tr[\sigma O^B(u)]=0$. Then, by \Cref{thm:mainbound}, we have 
\begin{align}\label{mainboundTIp}
\frac{\mathbb{P}\Big(\frac{1}{t}\int_0^t\tr[O^B(u)\rho_s]\,ds+\frac{W_t}{t} >r\Big)}{\|\Gamma_\sigma^{-1}(\rho)\|_{L^2(\sigma)}}
\, \le \,\exp\Big(-t\,\underset{\rho}{\inf}\,\big\{I_\cL(\rho)+\frac{1}{2}\,\big(r-\tr[\rho\,\widetilde{O}^B(u)]\big)^2\big\}\Big)\,,
\end{align}	
where $\widetilde{O}^B(u):=\Delta_\sigma^{\frac{1}{4}}(L^*_u)+\Delta_\sigma^{-\frac{1}{4}}(L_u)$. Next, by $\alpha(\cT_\Phi)\operatorname{I}$, we can further bound $I_\cL(\rho)$ as follows:
\begin{align}\label{eqTCfisherwasser}
I_\cL(\rho)\ge \frac{\cT_\Phi(\rho,\sigma)^2}{2C}\ge \frac{\tr[\rho \,\widetilde{O}^B(u)]^2}{2}\,,
\end{align}
where we used that $(\sqrt{C}\widetilde{O}^B(u),\sqrt{C}\widetilde{O}^B(u))\in\Phi$ and that $\tr[\sigma\,\widetilde{O}^B(u)]=\tr[\sigma\, O^B(u)]=0$. Therefore,
\begin{align*}
    \frac{\mathbb{P}\Big(\frac{1}{t}\int_0^t\tr[O^B(u)\,\rho_s]\,ds+\frac{W_t}{t} >r\Big)}{\|\Gamma_\sigma^{-1}(\rho)\|_{L^2(\sigma)}}
\, \le \exp\Big(-t\,\inf_\rho\,\frac{\tr[\rho \,\widetilde{O}^B(u)]^2}{2}+\frac{\big(r-\tr[\rho\,\widetilde{O}^B(u)]\big)^2}{2}\Big)
\le \exp\Big(-\frac{t\,r^2}{4}\Big)\,,
\end{align*}
where the last bound follows from the two-points inequality $2(a^2+b^2)\ge (a-b)^2$. The result follows.
\end{proof}

\begin{remark}
Observe that we do not assume $u\in S^{k-1}(\mathbb{C})$ above, but instead simply $u\in\mathbb{C}^k$. This is done up to a scaling of the Brownian motion $W_t\to \|u\| \,W_t$, where $\|\cdot\|$ is the canonical $2$-norm.
\end{remark}

\begin{corollary}\label{cor:lip}
With the notations of \Cref{theo:ti}, we have that, under $\operatorname{TI}(C)$, for any $u\in\mathbb{C}^k$ and any indirectly measured observable $O^B(u):=L_u+L^*_u$ and all $t,r>0$
\begin{align}\label{mainboundTIcoro}
\mathbb{P}\Big(\frac{1}{t}\int_0^t\tr[O^B(u)\,\rho_s]\,ds+\|u\|\,\frac{W_t}{t} >\tr[\sigma \,O^B(u)]+r\Big)\le     \|\Gamma_\sigma^{-1}(\rho)\|_{\mathbb{L}_{2}(\sigma)}\, \exp\left(-\frac{t\,r^2}{2\big(1+C\|\widetilde{O}^B(u)\|_{\operatorname{Lip}}^2\big)}\right)\,,
\end{align}	
where we recall that $\widetilde{O}^B(u)=\Delta_\sigma^{\frac{1}{4}}(L^*_u)+\Delta_\sigma^{-\frac{1}{4}}(L_u)$.
\end{corollary}

\begin{proof}
This directly follows from a simple modification of the proof of \Cref{theo:ti}: first, \eqref{eqTCfisherwasser} is replaced by the bound
\begin{align*}
    I_\cL(\rho)\ge \frac{W_{1,\cL}(\rho,\sigma)^2}{2C}\ge \frac{\tr[\rho \,\widetilde{O}^B(u)]^2}{2C\,\|\widetilde{O}^B(u)\|_{\operatorname{Lip}}^2}\,.
\end{align*}
The result follows after replacing the use of the inequality $2(a^2+b^2)\ge (a-b)^2$ by $\big[\big(\frac{a}{b}\big)^2+(r-a)^2\big]\ge \frac{r^2}{1+b^2}$, with $a=\tr[\rho\widetilde{O}^B(u)]$ and $b=\sqrt{C}\|\widetilde{O}^B(u)\|_{\operatorname{Lip}}$. 
\end{proof}

\paragraph{Concentration from Poincar\'{e} inequality}

We now prove the following weaker concentration bound depending on the gap of $\cL$ (see also Theorem 3.1 in \cite{Guillin2009}).

\begin{theorem}
Let $t\mapsto \e^{t\cL}$ be a finite dimensional, primitive, $\operatorname{KMS}$ symmetric quantum Markov semigroup. Then, for any initial state $\rho\in\cD(\cH)$, $t>0$, $u\in \mathbb{C}^{k}$ and any indirectly measured observable $O^B(u):=L_u+L^*_u$, 
\begin{align*}
	\mathbb{P}_\rho\Big(\,\frac{1}{t}\,\int_0^t\,\tr[O^B(u)\,\rho_s]\,ds+\|u\|\,\frac{W_t}{t}\ge\tr[\sigma\,O^B(u)]+ r\Big)\le\|\Gamma_\sigma^{-1}(\rho)\|_{L^2(\sigma)}\,\exp\left(-\frac{\lambda(\cL)\,t\,r^2}{2\big(\lambda(\cL)+2\|\widetilde{O}^B(u)\|_{\infty}^2\big)}\right)\,,
\end{align*}
where we recall that $\widetilde{O}^B(u)=\Delta_\sigma^{\frac{1}{4}}(L^*_u)+\Delta_\sigma^{-\frac{1}{4}}(L_u)$. and where $\lambda(\cL)$ is the gap of $\cL$.
	\end{theorem}
\begin{proof}
The proof follows the same lines as those of the proof of \Cref{cor:lip}, and is a direct consequence of \Cref{thm:mainbound}, \Cref{prop:poincartoTI} and the dual formulation of the trace norm.
		\end{proof}

\section{Examples}\label{examples}

\subsection{Depolarizing channel}
We first consider the simplest QMS, namely the depolarizing semigroup on $\cB(\cH)$, which is defined for any full-rank state $\sigma$ by
\begin{align}
    \e^{t\cL_\sigma}(X):=(1-\e^{-t})\id\tr[\sigma X]+\e^{-t}X\,.
\end{align}
One can readily check that the semigroup is primitive with unique invariant state $\sigma$, and that it is GNS symmetric with respect to $\sigma$. The LSI and MLSI constants of the semi-group $t\mapsto e^{t\cL_\sigma}$ were computed in \cite{MllerHermes2016,Beigi2020}. In particular, we have
\begin{align}
 &   \alpha_2(\cL_\sigma)=\frac{1-2 {s}_{\min}(\sigma)}{\ln\big(s_{\min}(\sigma)^{-1}-1\big)}\,,\qquad \qquad \alpha_2(\cL_{d^{-1}\id})=\frac{d-2}{d\ln(d-1)}\,.
\end{align}
where $s_{\min}(\sigma)$ denotes the minimal eigenvalue of $\sigma$. For sake of simplicity, we restrict ourselves to the case when $\sigma$ is the maximally mixed state on $\cH$: $\sigma=d^{-1}\id$, where $d$ stands for the dimension of $\cH$. In that case, the Lindblad operators of $\cL_\sigma$ can be chosen as $L_{xy}:=d^{-\frac{1}{2}}|x\rangle\langle y|$ for an arbitrary orthonormal basis $\{|x\rangle\}_x$ of $\cH$. In this case, any self-adjoint operator $O$ can be interpreted as an indirectly measured observable 
\begin{align}\label{OtoO(u)}
O\equiv O^B(u)=\sum_{x,y=1}^d\,\frac{u_{x,y}}{\sqrt{d}}\,\Big(|x\rangle\langle y|+|y\rangle\langle x|\Big)\,.
\end{align}
Without loss of generality, the basis $\{|x\rangle\}_x$ can be chosen as to diagonilize the observable $O$, so that the eigenvalues $\{O_x\}_x$ satisfy $O_x=\delta_{x,y}\,\frac{u_{x,y}}{\sqrt{d}}$. Therefore
\begin{align*}
    \|u\|^2=\sum_{x,y}|u_{x,y}|^2=d\,\sum_x\,O_x^2=d\|O\|_2^2\,.
\end{align*}
By \Cref{TIfromLSI}, we then have that the depolarizing semi-group satisfies the transport-information inequality with constant $C=8^{-1}\alpha_2(\cL_{\sigma})^{-2}$. Therefore, from \Cref{cor:lip}, we derive the following concentration for the quantum trajectories: 

\begin{corollary}\label{corodepol}
 In the notations of \Cref{cor:lip}, for any observable $O$ with eigenvalues $\{O_x\}_x$ and all $t,r>0$,
\begin{align*}
\mathbb{P}\Big(\frac{1}{t}\int_0^t\tr[O\,\rho_s]\,ds+\sqrt{d}\|O\|_2\,\frac{W_t}{t} >\frac{1}{d}\tr[ \,O]+r\Big)\le     d\, \exp\left(-\frac{2(d-2)^2\,t\,r^2}{4(d-2)^2+\sum_{x,y}(O_x-O_y)^2\,{d^2\ln(d-1)^2}}\right)\,.
\end{align*}	
\end{corollary}
\begin{proof}
In light of the paragraph before \Cref{corodepol}, it remains to express the Lipschitz constant of $O$. Since $\sigma=d^{-1}\id$, the constants $\{\omega_{x,y}\}_{x,y}$ in the definition of the Lipschitz constant are all equal to $0$. Without loss of generality, we can chose the basis $\{|x\rangle\}_x$ in which $O$ is diagonal. Then
\begin{align*}
    \|O\|_{\operatorname{Lip}}^2=\sum_{x,y}\,2\,\big\|\big[|x\rangle\langle y|,O\big]\|_\infty^2=2\,\sum_{x,y}\,(O_x-O_y)^2\,.
\end{align*}
The result follows. 
\end{proof}

\subsection{Tensorization}\label{sec:tensorization}
In the classical framework, when two semigroups each satisfy a (modified) logarithmic Sobolev inequality with constants $\alpha(\cL_1)$ and $\alpha(\cL_2)$ respectively, then their tensor product will also satisfy the same inequality with a constant $\alpha=\min(\alpha(\cL_1),\alpha(\cL_2))$. This property is known as \textit{dimension-free tensorization}. In the case of a transportation-cost or transportation-information inequality, the constant will rather scale as the sum of the constants of the local systems. Since logarithmic Sobolev inequalities imply transportation cost ones, the former provide a much stronger notion of tensorization when satisfied.

Tensorization is a much more subtle property to prove in the quantum realm. In general, neither the tensorization of the logarithmic Sobolev constant, nor that of the modified logarithmic Sobolev constant is known to hold for general classes of semigroups. Some exceptions for the former include the class of primitive qubit unital semigroups \cite{King2014} or the qubit depolarizing semigroup $t\mapsto \e^{t \cL_\sigma}$ \cite{Beigi2020}. Similarly, the modified logarithmic Sobolev constant is only known to tensorize for a few cases including e.g.,~the quantum Ornstein Uhlenbeck semigroup \cite{carlen2017gradient,DePalma2018}.

Thankfully, new techniques have been introduced in order to deal with the current lack of a proof of tensorization of the aforementioned quantum functional inequalities. In the case of the logarithmic Sobolev constant, the following was proved in \cite[Theorem 9]{[TPK14]}:
\begin{lemma}\label{lem:tensor}
Let $N\in \mathbb{N}$ and for any $k\in \{1,...,N\}$ let $t\mapsto \e^{t\cL_{k}}$
be a primitive $\operatorname{QMS}$ acting on $\cB(\cH)$ verifying $\operatorname{KMS}$ $\operatorname{QDB}$ with respective invariant state $\sigma_k$ and spectral gap $\lambda_k$. Then, the logarithmic Sobolev constant $\alpha_2(\cL^{(N)})$ of the product $\operatorname{QMS}$ $t\mapsto \e^{t\sum_k\cL_k\otimes \id_{k^c}}$ satisfies
\begin{align}\label{eq:LSIquasifactor}
\frac{\min_k\{\lambda_k\}}{\ln\big(d^4\max_k\{\|\sigma_k^{-1}\|_\infty\}\big)+11}\le     \alpha_2(\cL^{(N)})\le \frac{\min_k\{\lambda_k\}}{2}\,,
\end{align}
where $d:=\operatorname{dim}(\cH)$. In particular, the lower bound in \eqref{eq:LSIquasifactor} is independent of the number $N$ of subsystems. 
\end{lemma}
We also mention in passing the recent advances in proving tensorization of the modified logarithmic Sobolev constant beyond the primitive case in \cite{gao2021spectral}. \Cref{lem:tensor} can be used in combination with \Cref{cor:lip} to provide a tensorization result for the concentration bounds derived in \Cref{sec:concentration}:
\begin{corollary}
 Assume that the $\operatorname{QMS}$ $t\mapsto \e^{t\sum_{k=1}^n\cL_k\otimes\id_{k^c}}$ has $\operatorname{LSI}$ constant $\alpha_2$. Furthermore, we denote by $\{\omega_{k,j}\}_{j\in\mathcal{J}}$ the Bohr frequencies of $\cL_k$, and by $\{L_{k,j}\}_{j\in\mathcal{J}}$ its corresponding Lindblad operators. Then, for any observable 
  $O_{{u}}=\big(L_{{u}}+L_{{u}}^*\big)$ with $u=(u_{k,j})_{k,j}\in\mathbb{C}^{n|\cJ|}$, we have
 \begin{align}\label{mainboundTIcoro1}
 \mathbb{P}\Big(\frac{1}{t}\int_0^t\tr[O_{{u}}\rho_s]\,ds+\|u\|\,\frac{\,W_t}{t} >\tr[\otimes_k\sigma_k O_{{u}}]+r\Big)\le     \|\Gamma_{\otimes_k\sigma_k}^{-1}(\rho)\|_{L^{2}(\otimes_k\sigma_k)}\, \exp\left(-\frac{4\alpha_2^2\,t\,r^2}{8\alpha_2^2+n\alpha(u)}\right)\,,
 \end{align}	
 where $\alpha({u}):=2\,|\mathcal{J}| \max_{k,j}\e^{\omega_{k,j}/2}\|\sum_{i}u_{k,i}\big[ L_{k,j},\Delta_{\sigma_k}^{\frac{1}{4}}(L_{k,i}^*)+\Delta_{\sigma_k}^{-\frac{1}{4}}(L_{k,i})\big]\big\|_\infty^2$.
\end{corollary}
\begin{proof}
We denote $\sigma:=\otimes_k\sigma_k$. By \Cref{cor:lip}, we have that if the tensor product semigroup satisfies a transportation-information inequality then \eqref{mainboundTIcoro1} follows as long as the Lipschitz constant on the right-hand side of \eqref{mainboundTIcoro} is upper bounded by $\alpha({u})$. This follows from:
\begin{align}
    \|\widetilde{O}^B(u)\|_{\operatorname{Lip}}^2&=\sum_{k=1}^n \sum_{j\in\mathcal{J}}\,\big(\e^{-\omega_{k,j}/2}+\e^{\omega_{k,j}/2}\big)\,\big\|\big[ L_{k,j},\Delta_\sigma^{\frac{1}{4}}(L_{{u}}^*)+\Delta_\sigma^{-\frac{1}{4}}(L_{{u}})\big]\big\|_\infty^2\\
    &=\sum_{k=1}^n \sum_{j\in\mathcal{J}}\,\big(\e^{-\omega_{k,j}/2}+\e^{\omega_{k,j}/2}\big)\,\big\|\sum_{i}u_{k,i}\big[ L_{k,j},\Delta_{\sigma_k}^{\frac{1}{4}}(L_{k,i}^*)+\Delta_{\sigma_k}^{-\frac{1}{4}}(L_{k,i})\big]\big\|_\infty^2\\
     &\le  n 2\,|\mathcal{J}| \max_{k,j}\e^{\omega_{k,j}/2}\|\sum_{i}u_{k,i}\big[ L_{k,j},\Delta_{\sigma_k}^{\frac{1}{4}}(L_{k,i}^*)+\Delta_{\sigma_k}^{-\frac{1}{4}}(L_{k,i})\big]\big\|_\infty^2\,.
\end{align}
The result then follows from \Cref{TIfromLSI} which establishes that transport-information is implied by LSI. 
\end{proof}

\subsection{Gibbs samplers}

Tensorization can be thought of as a property of non-interacting systems or of systems at infinite temperature,
for which the evolution can be written as a tensor power of local channels. Proofs of LSI/MLSI
for quantum interacting spin systems (a.k.a. Gibbs samplers) have recently attracted the attention of the community \cite{Bardet2021,Bardet2021b,capel2021modified}. More recently, it was shown that Gibbs states over arbitrary graphs satisfy a transportation cost inequality at large enough temperature \cite{de2021quantum}. Building on the techniques of \cite{de2021quantum}, one of the authors proves in an article to appear that a certain class of Gibbs samplers satisfies the transportation cost-information inequality at high enough temperature. 

More precisely, let $G=(V,E)$ be a graph with $n=|V|$, and let $\cH_V:=\bigotimes_{v\in V}\cH_v$ be the Hilbert space of a local quantum system, with $\cH_v:=\mathbb{C}^d$ for all $v\in V$. The interactions are modeled through the Hamiltonian $H:=\sum_{A\subset \Lambda}h_A\otimes 1_{A^c}$, where each local self-adjoint operator $h_A$ satisfies $\|h_A\|\le 1$ and is supported on the region $A\subset V$. Here, we also assume that the Hamiltonian is of finite-range, which means that the size and diameter of any region $A$ appearing in the decomposition of $H$ is uniformly bounded by a constant $r>0$ independent of $n$. We also assume the interaction to be commuting, i.e.,~$[h_A,h_{A'}]=0$ for all $A,A'$. Next, we define the Gibbs state $\omega$ associated to $H$ at inverse temperature $\beta>0$ as
\begin{align}
    \omega:=\frac{e^{-\beta H}}{\tr\big[e^{-\beta H}\big]}\,.
\end{align}
A Gibbs sampler is a locally defined quantum channel which prepares an approximation of the Gibbs state $\omega $ starting from any initial state on $\cH_V$. The efficiency of the Gibbs sampler depends on the time it takes to reach the approximating state. Here, we consider the heat-bath generator which is defined as follows: for a given site $v\in V$, we denote the composition of the partial trace $\tr_v$ on $v$ with the Petz recovery map of $v$ as 
\begin{align}
    \Psi_v^* (\rho)=\Phi_v^*\circ \tr_v(\rho)=\omega^{\frac{1}{2}}\omega_{v^c}^{-\frac{1}{2}}\,(\rho_{v^c}\otimes I_v)\,\omega_{v^c}^{-\frac{1}{2}}\omega^{\frac{1}{2}}\,,
\end{align}
where $\omega$ is the Gibbs state of the Hamiltonian $H$, and where we denoted by $\omega_A$ the reduced state on the subregion $A\subseteq V$. Clearly, when $H$ is made of commuting terms, the map $\Psi_v$ acts non-trivially on the neighborhood of $v$, which is defined as
\begin{equation}
    N_v:=\bigcup\left\{A\in E : v\in A\right\}\,.
\end{equation}
Next, we introduce the generator of the heat-bath dynamics
\begin{align}
\cL_V:=\sum_{v\in V}\cL_v\,,
\end{align}
where $\cL_v:=\Psi_v-\id$. The quantum Markov semigroup $t\mapsto e^{t\cL_V^*}$ generated by $\cL_V^*$ converges to $\omega$ as $t\to\infty$. The following result is a direct consequence of  \cite[Proposition 4.13 and Proposition 6.1]{GaoRouzecurvature}.
\begin{lemma}\label{TIGibbs}
There exists an inverse temperature $\beta_c>0$ such that, for any $\beta<\beta_c$, the $\operatorname{QMS}$ $t\mapsto e^{t\cL_V}$ satisfies the transportation cost-information inequality with respect to the quantum Ornstein distance with constant $C=\mathcal{O}(n)$: for any quantum state $\rho$,
\begin{align*}
    W_1^{\operatorname{Orn}}(\rho,\omega)\le \sqrt{2C\,I_{\cL_V}(\rho)}\,.
\end{align*}
\end{lemma}
In what follows, we denote again by $k$ the number of Linblad operators contained in $\mathcal{L}_V$. Clearly, $k=\mathcal{O}(|V|)$. Combining the previous Lemma with the main concentration bound \eqref{mainboundTIp_gauss}, we find the following corollary:

\begin{corollary}
For any initial state $\rho\in\cD(\cH)$, $t>0$, $u\in \mathbb{C}^k$, and any indirectly measured observable $O^B(u):=L_{u}+L^*_u$,
\begin{align}\label{mainboundTIp_gauss1}
\mathbb{P}\Big(\frac{1}{t}\int_0^t\tr[O^B(u)\rho_s]\,ds+\|u\|\,\frac{W_t}{t} >\tr[\omega O^B(u)]+r\Big)\le   \, \exp\left(\frac{\beta\|H\|_\infty}{2}-\frac{t\,r^2}{2\big(1+C(\|\widetilde{O}^B(u)\|_{\operatorname{Lip}}^{\operatorname{Orn}})^2\big)}\right)\,,
\end{align}		
where $C=\mathcal{O}(n)$ was introduced in \Cref{TIGibbs}, and where $\widetilde{O}^B(u):=\Delta_\sigma^{\frac{1}{4}}(L^*_u)+\Delta_\sigma^{-\frac{1}{4}}(L_u)$.
\end{corollary}

\begin{proof}
The result follows directly from  \Cref{cor:lip} up to the replacement of $\|.\|_{\operatorname{Lip}}$ to $\|.\|_{\operatorname{Lip}}^{\operatorname{Orn}}$ and the estimate:
\begin{align*}
    \|\Gamma_\sigma^{-1}(\rho)\|_{L^{2}(\sigma)}=\big\|e^{\frac{\beta H}{4}}\rho e^{\frac{\beta H}{4}}\big\|_2\le e^{\frac{\beta\|H\|_\infty}{2}}\,.
\end{align*}
\end{proof}

Another Gibbs sampler which models the thermalization of a quantum system weakly interacting with a large reservoir is the so-called Davies dynamics \cite{Davies1980}. 
In \cite{Kastoryano2014}, it is proved that Davies dynamics are gapped at any inverse temperature $\beta>0$ in 1D and on regular lattices below a threshold inverse temperature $\beta_c>0$. This result was extended to the MLSI in the 1D case in a paper to appear \cite{MLSI1D}. However, it is still open whether these dynamics also satisfy a transportation cost-information inequality under reasonable assumptions.

\paragraph{Acknowledgement} The research of T.B. has been supported by ANR project ESQUISSE (ANR-20-CE47-0014-01) of
the French National Research Agency (ANR). The research of
T.B. and C.R. has been supported by ANR project QTraj (ANR-20-CE40-0024-01) of
the French National Research Agency (ANR). C.R. acknowledges the support of the Munich Center for Quantum Sciences and Technology, as well as the Humboldt Foundation. The research of L.H. was partially supported by the Swiss National Science Foundation (Grant No. P2EZP2-188093).

\appendix
\section{The quantum perturbed generators}\label{app:sec:compgen}

Here we present the computational steps for the derivation of the generators $\mathcal{L}_{i\lambda,u}^B$, $\mathcal{L}_{i\lambda,u}^P$, and $\mathcal{L}_{i\lambda,\bu}$ in \Cref{sec:fksg} for single-mode Brownian motion (cf.\ \Cref{subsubsec:BM}), single-mode Poisson processes (cf.\ \Cref{subsubsec:PP}), and the general multi-mode case with both kinds of processes included (cf.\ \Cref{subsubsec:BMm}), respectively.

As it will be used in the derivations below, we recall the following product rule~\cite{Parthasarathy}:
Let $\mathbf{E}_t,\mathbf{F}_t,\mathbf{G}_t,\mathbf{H}_t$ and $\mathbf{E}'_t,\mathbf{F}'_t,\mathbf{G}'_t,\mathbf{H}'_t$ be appropriate families of operators on $\cH_S\otimes \Gamma$, and let $X_t, X'_t$ be defined by the differential equations 
\begin{align*}
\dd X_t&=\mathbf{E}_t.\dd \mathbf{A}_{\mathbf{u}}(t)+\mathbf{F}_t. \mathbf{A}_{\mathbf{u}}^*(t)+\mathbf{G}_t .\dd \mathbf{\Lambda}_{\mathbf{u}\mathbf{u}^*}(t)+\mathbf{H}_t \dd t\,,\\
\dd X'_t&=\mathbf{E}'_t.\dd \mathbf{A}_{\mathbf{u}}(t)+\mathbf{F}'_t. \mathbf{A}_{\mathbf{u}}^*(t)+\mathbf{G}'_t. \dd \mathbf{\Lambda}_{\mathbf{u}\mathbf{u}^*}(t)+\mathbf{H}'_t \dd t\,.
\end{align*}
Then
\begin{align}\label{eq:qscprodrule}
\dd (X_tX'_t)=(\dd X_t) X'_t+X_t (\dd X'_t)+(\dd X_t)(\dd X'_t)\,,
\end{align}
where 
\begin{align*}
(\dd X_t)X'_t=\mathbf{E}_t X'_t.\dd \mathbf{A}_{\mathbf{u}}(t)+\mathbf{F}_t X'_t.\dd \mathbf{A}_{\mathbf{u}}^*(t)+\mathbf{G}_t X'_t.\dd \mathbf{\Lambda}_{\mathbf{u}\mathbf{u}^*}(t)+\mathbf{H}_t X'_t\dd t\,, 
\end{align*}
and $(\dd X_t).(\dd X'_t)$ is evaluated according to the quantum It\^o rules~\eqref{eq:noncommito}. Furthermore, the following observations will be used extensively throughout this section:
\renewcommand{\labelenumi}{(\alph{enumi})}
\renewcommand{\theenumi}{(\alph{enumi})}
\begin{enumerate}
    \item \label{obs:commopdistsupp} Operators with distinct support in $\mathbb R_+$ commute.
    \item \label{obs:commdiff} Quantum noises infinitesimal increments at time $t$ commute with quantum adapted processes at time $t$ \cite{Parthasarathy}.
\end{enumerate}

\subsection{Single-mode Brownian motion}\label{app:subsec:compgenbm}
Consider the semigroup $t\mapsto {\Phi_{t,u}^B}^{(i\lambda)}$ on $M_n(\mathbb{C})$ defined in \Cref{def:perturbedsemigroupsmbm} as
\begin{align*}
{\Phi_{t,u}^B}^{(i\lambda)}(X)=\tr_\Gamma\left(\Omega\ {V_{t,u}^B}^*(-\lambda)XV_{t,u}^B(\lambda)\right)\,,
\end{align*}
with $V_{t,u}^B(\lambda)= \e^{i\frac{\lambda}{2}(A_{u}(t)+A_{u}^*(t))}U_t$. For the derivation of its generator $\mathcal{L}_{i\lambda,u}^B$, we will compute the quantum stochastic differential of the process ${V_{t,u}^B}^*(-\lambda)XV_{t,u}^B(\lambda)$. Let us start by differentiating the family of operators ${V_{t,u}^B}^*(-\lambda)$, $t\geq 0$: Using the product rule~\eqref{eq:qscprodrule} we get 
\begin{align}
\dd {V_{t,{u}}^B}^*(-\lambda)
=& \dd U_t^*\e^{i\frac{\lambda}{2}(A_{u}(t)+A_{u}^*(t))}+ U_t^*\dd\e^{i\frac{\lambda}{2}(A_{u}(t)+A_{u}^*(t))}+ \dd U_t^*\dd\e^{i\frac{\lambda}{2}(A_{u}(t)+A_{u}^*(t))}\nonumber\\
=& U_t^*\left(\left(iH-\frac{1}{2}\bL^*.\bL\right) \d t+\bL^*.\d \bA_{\be}(t)- \bL.\d \bA_{\be}^*(t)\right)\e^{i\frac{\lambda}{2}(A_{u}(t)+A_{u}^*(t))}\nonumber\\
&+ U_t^*\e^{i\frac{\lambda}{2}(A_{u}(t)+A_{u}^*(t))}\left(\frac{i\lambda}{2}\dd (A_{u}(t)+A_{u}^*(t))+\left(\frac{i\lambda}{2}\right)^2\frac{1}{2}\dd t\right)\nonumber\\
&+ U_t^*\left(\left(iH-\frac{1}{2}\bL^*.\bL\right) \d t+\bL^*.\d \bA_{\be}(t)- \bL.\d \bA_{\be}^*(t)\right)\e^{i\frac{\lambda}{2}(A_{u}(t)+A_{u}^*(t))}\nonumber\\
&\qquad \qquad \qquad \left(\frac{i\lambda}{2}\dd (A_{u}(t)+A_{u}^*(t))+\left(\frac{i\lambda}{2}\right)^2\frac{1}{2}\dd t\right)\nonumber\\
=& {V_{t,{u}}^B}^*(-\lambda)\Biggl(\left(iH-\frac{1}{2}\bL^*.\bL+\left(\frac{i\lambda}{2}\right)^2\frac{1}{2}+\frac{i\lambda}{2}L^*_u\right)\dd t\nonumber\\
&\qquad\qquad\qquad+\left(\bL^*+\frac{i\lambda}{2}u\right).\dd \bA_{\be}(t)-\left(\bL-\frac{i\lambda}{2}u\right).\dd \bA_{\be}^*(t)\Biggr)\,.\label{eq:diffvb}
\end{align}
Here we used the defining differential equation~\eqref{eq:diffeqU} of $U_t$ and the quantum It\^o rules~\eqref{eq:noncommito} for the second equality, and the identities  $A_{u}(t)+A_{u}^*(t)=u.(\bA_{\be}(t)+\bA_{\be}^*(t))$, $L_u=u.\bL$, together with the observations~\ref{obs:commopdistsupp} and \ref{obs:commdiff}, for the last equation. Applying again the product rule~\eqref{eq:qscprodrule} and subsequently plugging in \Cref{eq:diffvb}, we thus get
\begin{align}
\dd \left({V_{t,{u}}^B}^*(-\lambda) X {V_{t,{u}}^B}(\lambda)\right)
=&\dd {V_{t,{u}}^B}^*(-\lambda)X{V_{t,{u}}^B}(\lambda)+{V_{t,{u}}^B}^*(-\lambda)X\dd {V_{t,{u}}^B}(\lambda)+\dd {V_{t,{u}}^B}^*(-\lambda)X\dd {V_{t,{u}}^B}(\lambda)\nonumber\\
=&{V_{t,{u}}^B}^*(-\lambda)\left(\left(iH-\frac{1}{2}\bL^*.\bL+\left(\frac{i\lambda}{2}\right)^2\frac{1}{2}+\frac{i\lambda}{2}L^*_u\right)\dd t\right)X{V_{t,{u}}^B}(\lambda)\nonumber\\
&+{V_{t,{u}}^B}^*(-\lambda)\left(\left(\bL^*+\frac{i\lambda}{2}u\right).\dd \bA_{\be}(t)-\left(\bL-\frac{i\lambda}{2}u\right).\dd \bA_{\be}^*(t)\right)X{V_{t,{u}}^B}(\lambda)\nonumber\\
&+{V_{t,{u}}^B}^*(-\lambda)X\left(\left(-iH-\frac{1}{2}\bL^*.\bL+\left(\frac{i\lambda}{2}\right)^2\frac{1}{2}+\frac{i\lambda}{2}L_u\right)\dd t\right) {V_{t,{u}}^B}(\lambda)\nonumber\\
&+{V_{t,{u}}^B}^*(-\lambda)X\left(\left(\bL+\frac{i\lambda}{2}u\right).\dd \bA_{\be}^*(t)-\left(\bL^*-\frac{i\lambda}{2}u\right).\dd \bA_{\be}(t)\right) {V_{t,{u}}^B}(\lambda)\nonumber\\
&+{V_{t,{u}}^B}^*(-\lambda)\left(\left(\bL^*+\frac{i\lambda}{2}u\right).\dd \bA_{\be}(t)\right)X\left(\left(\bL+\frac{i\lambda}{2}u\right).\dd \bA_{\be}^*(t)\right){V_{t,{u}}^B}(\lambda)\label{eq:Xt}\\
=& {V_{t,{u}}^B}^*(-\lambda)\left(\left(i\lambda(L^*_u X+XL_u )-\frac{\lambda^2}{2}X+i[H,X]-\frac{1}{2}\{\bL^*.\bL,X\}+\bL^*X.\bL\right)\dd t\right){V_{t,{u}}^B}(\lambda)\nonumber\\ 
&+ {V_{t,{u}}^B}^*(-\lambda)\left(\left([\bL^*,X]+i\lambda u X\right).\dd \bA_{\be}-\left([\bL,X]-i\lambda u X\right).\dd \bA_{\be}^*\right){V_{t,{u}}^B}(\lambda)\label{eq:Xt2}\,.
\end{align}
Note that we only kept the terms yielding non-zero contributions according to the quantum It\^o rules~\eqref{eq:noncommito} in line \eqref{eq:Xt}, and we used again the observations~\ref{obs:commopdistsupp} and \ref{obs:commdiff} for the last equation~\eqref{eq:Xt2}. We conclude by observing that, for $\ket{\Omega}\bra{\Omega}\coloneqq \Omega$ and arbitrary pure states $\ket{w}\bra{w},\ket{v}\bra{v}$ on $S$, we have
\begin{align}
\BraKet{v}{\left({\Phi_{t,u}^B}^{(i\lambda)}(X)-X \right)w}&=\BraKet{v\otimes \Omega}{\left({V_{t,u}^B}^*(-\lambda)X{V_{t,u}^B}(\lambda)-X \right)w\otimes \Omega}\nonumber\\
&=\int_0^t \BraKet{v\otimes \Omega}{{V_{s,u}^B}^*(-\lambda)\mathcal{L}_{i\lambda,u}^B(X){V_{s,u}^B}(\lambda)w\otimes \Omega}d s\label{eq:matrixel2}\\
&=\int_0^t \BraKet{v\otimes \Omega}{{\Phi_{s,u}^B}^{(i\lambda)}\left(\mathcal{L}_{i\lambda,u}^B(X)\right)w\otimes \Omega}d s\,,\nonumber
\end{align}
where we inserted the differential computed in~\Cref{eq:Xt2} for the second equation; \Cref{eq:matrixel2} then follows from the fact that the quantum stochastic differentials vanish when applied to the vacuum state (cf.\ Section~\ref{sec:noncommnoises}). We have thus proven the stated form of the generator $\mathcal{L}_{i\lambda,u}^B$ of the semigroup $t\mapsto {\Phi_{t,u}^B}^{(i\lambda)}$ in \Cref{def:gensmbm}.

\subsection{Single-mode Poisson process}\label{app:subsec:compgenpoisson}
Consider now the semigroup $t\mapsto {\Phi_{t,u}^P}^{(i\lambda)}$ on $M_n(\mathbb{C})$ defined in \Cref{def:perturbedsemigroupsmpp} as
\begin{align*}
{\Phi_{t,u}^P}^{(i\lambda)}(X)=\tr_\Gamma\left(\Omega\ {V_{t,u}^P}^*(-\lambda)XV_{t,u}^P(\lambda)\right)\,,
\end{align*}
with $V_{t,u}^P(\lambda)= \e^{i\frac{\lambda}{2}\Lambda_{u}(t)}U_t$, instead. Here we aim to derive the form of its generator $\mathcal{L}_{i\lambda,u}^P$ stated in \Cref{def:gensmpp}. Following the same strategy as in the previous section treating single-mode Brownian motion (cf.\ \Cref{app:subsec:compgenbm}), we first differentiate the family of operators ${V_{t,u}^P}^*(-\lambda)$, $t\geq 0$, in order to compute the differential of the process ${V_{t,u}^P}^*(-\lambda)XV_{t,u}^P(\lambda)$. We will use the fact that according to~\cite[Example 25.16]{Parthasarathy},
$$\dd \e^{\frac{i}{2}\lambda\Lambda_u(t)}=(\e^{\frac{i}{2} \lambda}-1)\e^{\frac{i}{2}\lambda\Lambda_u(t)}\dd \Lambda_u(t).$$
We here provide an heuristic proof of this fact. By the quantum It\^o rules~\eqref{eq:noncommito}, we have $(\dd \Lambda_u(t))^2=\dd \Lambda_u(t)$. Thus
\begin{align*}
\dd (\Lambda_u(t))^n=\sum_{k=1}^n\binom{n}{k}\Lambda_u(t)^{n-k}\dd \Lambda_u(t)\,,
\end{align*}
where we used the product rule~\eqref{eq:qscprodrule}, as well as observation~\ref{obs:commdiff}. With the same arguments, we get
\begin{align*}
\dd \e^{\frac{i}{2}\lambda \Lambda_u(t)}=&\dd \sum_{n=0}^{\infty}\frac{1}{n!}\left(\frac{i}{2}\lambda\right)^n (\Lambda_u(t))^n\\
=&\sum_{n=1}^{\infty}\sum_{k=1}^n\frac{1}{(n-k)!k!}\left(\frac{i}{2}\lambda\right)^n \Lambda_u(t)^{n-k}\dd \Lambda_u(t)\\
=&\sum_{k=1}^{\infty}\frac{1}{k!}\left(\frac{i}{2}\lambda\right)^k\sum_{n=k}^{\infty}\frac{1}{(n-k)!}\left(\frac{i}{2}\lambda\right)^{n-k} \Lambda_u(t)^{n-k}\dd \Lambda_u(t)\\
=&\sum_{k=1}^{\infty}\frac{1}{k!}\left(\frac{i}{2}\lambda\right)^k\e^{\frac{i}{2}\lambda \Lambda_u(t)}\dd \Lambda_u(t)\\
=&\left(\e^{\frac{i}{2}\lambda}-1\right)\e^{\frac{i}{2}\lambda \Lambda_u(t)}\dd \Lambda_u(t).
\end{align*}
Using this expression of the infinitesimal increment, observation~\ref{obs:commopdistsupp} and the identity $L_u=u.\bL$, we obtain
\begin{align}
\dd {V_{t,u}^P}^*(-{\lambda})
=& \dd U_t^*\e^{\frac{i}{2}{\lambda}{\Lambda}_{u}(t)}+ U_t^*\dd\e^{\frac{i}{2}{\lambda}{\Lambda}_{u}(t)}+ \dd U_t^*\dd\e^{\frac{i}{2}{\lambda}{\Lambda}_{u}(t)}\nonumber\\
=& U_t^*\left(\left(iH-\frac{1}{2}\bL^*.\bL\right) \d t+\bL^*.\d \bA_{\be}(t)- \bL.\d \bA_{\be}^*(t)\right) \e^{\frac{i}{2}{\lambda}{\Lambda}_{u}(t)}\nonumber\\
&+ U_t^*\e^{\frac{i}{2}{\lambda}{\Lambda}_{u}(t)}\left(\e^{i\frac{\lambda}{2}}-1\right)\dd \Lambda_u(t)\nonumber\\
&+ U_t^*\left(\left(iH-\frac{1}{2}\bL^*.\bL\right) \d t+\bL^*.\d \bA_{\be}(t)- \bL.\d \bA_{\be}^*(t)\right) \e^{\frac{i}{2}{\lambda}{\Lambda}_{u}(t)}\left(\e^{i\frac{\lambda}{2}}-1\right)\dd \Lambda_u(t)\nonumber\\
=& {V_{t,u}^P}^*(-{\lambda})\biggl(\left(iH-\frac{1}{2}\bL^*.\bL\right) \d t+\left(\left(\e^{i\frac{\lambda}{2}}-1\right)L_u^*\, u+\bL^*\right).\d \bA_{\be}(t)\nonumber\\
\qquad \qquad \qquad &+\left(\e^{i\frac{\lambda}{2}}-1\right)\d \Lambda_u(t)-\bL.\d \bA_{\be}^*(t)\biggr)\,.\label{eq:diffvp}
\end{align}
Hence, once more applying the product rule~\eqref{eq:qscprodrule}, and subsequently plugging in \Cref{eq:diffvp}, yields
\begin{align}
\dd \left({V_{t,{u}}^P}^*(-\lambda) X {V_{t,{u}}^P}(\lambda)\right)
=&\dd {V_{t,{u}}^P}^*(-\lambda)X{V_{t,{u}}^P}(\lambda)+{V_{t,{u}}^P}^*(-\lambda)X\dd {V_{t,{u}}^P}(\lambda)+\dd {V_{t,{u}}^P}^*(-\lambda)X\dd {V_{t,{u}}^P}(\lambda)\nonumber\\
=&{V_{t,{u}}^P}^*(-\lambda)\left(\left(iH-\frac{1}{2}\bL^*.\bL\right)\d t+\left(\e^{i\frac{\lambda}{2}}-1\right)\d \Lambda_u(t)\right)X{V_{t,{u}}^P}(\lambda)\nonumber\\
&+{V_{t,{u}}^P}^*(-\lambda)\left(\left(\left(\e^{i\frac{\lambda}{2}}-1\right)L_u^*\, u+\bL^*\right).\d \bA_{\be}(t)-\bL.\d \bA_{\be}^*(t)\right)X{V_{t,{u}}^P}(\lambda)\nonumber\\
&+{V_{t,{u}}^P}^*(-\lambda)X\left(\left(-iH-\frac{1}{2}\bL^*.\bL\right)\d t+\left(\e^{i\frac{\lambda}{2}}-1\right)\d \Lambda_u(t)\right) {V_{t,{u}}^P}(\lambda)\nonumber\\
&+{V_{t,{u}}^P}^*(-\lambda)X\left(\left(\left(\e^{i\frac{\lambda}{2}}-1\right)L_u\, u+\bL\right).\d \bA_{\be}^*(t)-\bL^*.\d \bA_{\be}(t)\right) {V_{t,{u}}^P}(\lambda)\nonumber\\
&+{V_{t,{u}}^P}^*(-\lambda)\left(\left(\e^{i\frac{\lambda}{2}}-1\right)\d \Lambda_u(t)+\left(\left(\e^{i\frac{\lambda}{2}}-1\right)L_u^*\, u+\bL^*\right).\d \bA_{\be}(t)\right)\label{eq:Yt}\\
&\quad X\left(\left(\e^{i\frac{\lambda}{2}}-1\right)\d \Lambda_u(t)+\left(\left(\e^{i\frac{\lambda}{2}}-1\right)L_u\, u+\bL\right).\d \bA_{\be}^*(t)\right){V_{t,{u}}^P}(\lambda)\label{eq:Yt2}\\
=& {V_{t,{u}}^P}^*(-\lambda)\left(\left(\left(\e^{i\lambda}-1\right)L_u^* X L_u+i[H,X]-\frac{1}{2}\{\bL^*.\bL,X\}+\bL^*X.\bL\right)\dd t\right){V_{t,{u}}^P}(\lambda)\label{eq:Yt3}\\
&+ {V_{t,{u}}^P}^*(-\lambda)\left(\e^{i\lambda}-1\right)\left( X \d \Lambda_u(t) +L_u^* X \d A_{u}(t)+X L_u \d A_{u}^*(t)\right){V_{t,{u}}^P}(\lambda)\,,
\end{align}
where we only kept terms yielding a non-zero contribution according to the quantum It\^o rules~\eqref{eq:noncommito} in lines~\eqref{eq:Yt} and \eqref{eq:Yt2}, and used observation~\ref{obs:commdiff} for the last equality. The form of the generator $\cL_{i\lambda,u}^P$ of the semigroup $t\mapsto {\Phi_{t,u}^P}^{(i\lambda)}$ now follows analogously to the case of single-mode Brownian motion (cf.\ \Cref{app:subsec:compgenbm}); in particular, only line \eqref{eq:Yt3} yields non-vanishing contributions.

\subsection{Multi-mode mixed process}\label{app:subsec:compgenmm}
In the most general case of a multi-mode mixed process, consider the semigroup $t\mapsto {\Phi_{t,\bu}}^{(i\lambda)}(X)$ on $M_n(\mathbb{C})$ defined in~\eqref{def:perturbedsemigroupmm} as
\begin{align*}
{\Phi_{t,\bu}}^{(i\lambda)}(X)=\tr_\Gamma\left(\Omega\ {V_{t,\bu}}^*(-\lambda)XV_{t,\bu}(\lambda)\right)\,,
\end{align*}
with $V_{t,\bu}(\lambda)= \e^{\frac{i}{2}\lambda.(\bA_{\bu^B}(t)+\bA_{\bu^B}^*(t),\Lambda_{\bu^P}(t))}U_t$ unitary. Recall that $\mathbf{u}\equiv(\mathbf{u}^B,\mathbf{u}^P)$, with $\mathbf{u}^B =\{u_1,\dots,u_q\}$ and $\mathbf{u}^P =\{u_{q+1},\dots,u_\ell\}$, is a family of orthogonal, normalized vectors in $\mathbb{C}^k$. This in particular implies that any two quantum stochastic processes indexed by $u_j$ and $u_k$ respectively, commute for $j\not=k$. Consequently, we may write $V_{t,\bu}(\lambda)$ as a product of commuting exponentials, each of which either corresponds to a single-mode Brownian motion, or to a single-mode Poisson process. The derivation of the form of the generator $\cL_{i\lambda, \bu}$ of the semigroup $t\mapsto {\Phi_{t,\bu}}^{(i\lambda)}(X)$ given in \Cref{def:genmm} is then a straight-forward consequence of the results derived in the two previous Sections~\ref{app:subsec:compgenbm} and \ref{app:subsec:compgenpoisson}.

\section{BKM and KMS QDB}\label{app:BKMKMS}

In this section we provide an example showing that BKM QDB and KMS QDB are incomparable. It is a refinement of the example introduced in \cite[Appendix B]{carlen2017gradient}.

Let $\sigma\in \mathcal D$ be of full rank and $\mathcal M:X\mapsto\int_0^1\sigma^{1-s}X\sigma^s \d s$. Then,
$$\mathcal M^{-1}:X\mapsto \int_0^\infty \frac{1}{\sigma+t}X \frac{1}{\sigma+t}\d t.$$
We recall that $\Gamma:X\mapsto \sigma^{\frac12}X\sigma^{\frac12}$. Then a completely positive unital map, or quantum channel, $\Phi$ with invariant state $\sigma$, meaning $\Phi^*(\sigma)=\sigma)$, is BKM symmetric if and only if
\begin{align}\label{eq:BKM_intertwine}
    \mathcal M^{-1}\circ\Phi^*=\Phi\circ\mathcal M^{-1}.
\end{align}
Similarly $\Phi$ is KMS symmetric if and only if
\begin{align}\label{eq:KMS_intertwine}
    \Phi^*\circ\Gamma=\Gamma\circ\Phi.
\end{align}
Then we have the following lemma.
\begin{lemma}\label{lem:BKMKMS_comute}
If a quantum channel $\Phi$ verifies both $\operatorname{BKM}$ and $\operatorname{KMS}$ symmetries, then
$$\mathcal M^{-1}\circ\Gamma\circ\Phi=\Phi\circ\mathcal M^{-1}\circ\Gamma.$$
\end{lemma}
Remark that $\mathcal M^{-1}\circ\Gamma=g(\Delta)$ with $g:\mathbb R_+^*\to\mathbb R_+^*;  x\mapsto\frac{\frac12\ln x}{\sinh(\frac12\ln x)}$ a continuous function such that $g(1/x)=g(x)$. Hence, each eigenvalue of $g(\Delta)$ different from $1$ is twice more degenerate than one of its antecedent by $g$. Furthermore since GNS symmetry implies commutation with $\Delta$ (see~\cite{fagnola2007generators}) we directly have that GNS symmetry implies KMS and BKM symmetries.

Let $(u_1,u_2)$ be an orthonormal basis of $\mathbb C^2$ with $u_1$ and $u_2$ real. Let $(v_1,v_2)$ be a basis of $\mathbb C^2$ with $v_1$ and $v_2$ real, normalized and such that $\langle v_1,v_2\rangle\neq0$. Let $K_1=|v_1\rangle \langle u_1|$ and $K_2=|v_2\rangle\langle u_2|$. Then $\Phi:X\mapsto K_1^*XK_1+K_2^*XK_2$ is a quantum channel. Following \cite[Appendix B]{carlen2017gradient} its invariant state is
\begin{align}\label{eq:expr_sigma}
    \sigma=\frac{a}{a+b}|v_1\rangle\langle v_1|+\frac{b}{a+b}|v_2\rangle\langle v_2|
\end{align}
with $a=\langle v_2,u_1\rangle^2$ and $b=\langle v_1,u_2\rangle^2$. Since $v_1$ and $v_2$ are not orthogonal, $2>a+b>0$. Assume $a+b\neq 1$, then the spectrum of $\Phi$ is $\{1,1-a-b,0\}$ with $0$ twice degenerate. The kernel of $\Phi^*$ is thus $\operatorname{linspan}\{|u_1\rangle\langle u_2|,|u_2\rangle\langle u_1|\}$ and $\Phi$ is primitive as long as $ab\neq0$.

Let $\Psi=\Phi^{\operatorname{KMS}}\circ\Phi$. By construction $\Psi$ is irreducible and has $\sigma$ as its unique invariant state. Moreover $\Psi$ is KMS symmetric. We now show that in general $\Psi$ does not commute with $g(\Delta)$ and therefore is not BKM symmetric. Let us chose the vectors $u_1,u_2, v_1$ and $v_2$ such that $a+b\neq1$, $ab\neq 0$, $a\neq b$ and $\sigma\neq \id/2$. Let $\sigma=\lambda |\eta_1\rangle\langle \eta_1|+(1-\lambda)|\eta_2\rangle\langle \eta_2|$ be the spectral decomposition of $\sigma$ with $\lambda\in ]0,1/2[$ and $\eta_1$ and $\eta_2$ real. It follows that $g(\Delta)$ has two eigenspaces, $E_1=\operatorname{linspan}\{|\eta_1\rangle\langle\eta_1|, |\eta_2\rangle\langle \eta_2|\}$ associated with the eigenvalue $1$ and $E_\lambda=\operatorname{linspan}\{|\eta_1\rangle\langle\eta_2|, |\eta_2\rangle\langle \eta_1|\}$ associated with the eigenvalue $g(\lambda/(1-\lambda))$. 

Assume $\Psi$ commutes with $g(\Delta)$. Then, from the spectral decomposition of $g(\Delta)$,
$$\Psi(E_1)\subset E_1\quad\mbox{and}\quad \Psi(E_\lambda)\subset E_\lambda.$$
First, $\id\in E_1$. Let $X_0=(1-\lambda)|\eta_1\rangle\langle \eta_1|-\lambda|\eta_2\rangle\langle\eta_2|$. Then, $X_0\in E_1$ and $\langle X_0,\id\rangle_{KMS}=0$. Since $\id$ is an eigenvector of $\Psi$, $\Psi$ is KMS symmetric and $\Psi(E_1)\subset E_1$, $X_0$ is an eigenvector of $\Psi$. Assume $X_0$ is an element of the kernel of $\Psi$. It follows from the definition of $\Psi$ that it is also an element of the kernel of $\Phi$. That implies
$$(1-\lambda)\langle \eta_1,v_i\rangle^2=\lambda\langle\eta_2,v_i\rangle^2$$
for $i=1,2$. Thus $\langle\eta_1,v_1\rangle^2=\langle \eta_1,v_2\rangle ^2=\lambda$ and up to an irrelevant sign $|v_1\rangle=\sqrt{\lambda}|\eta_1\rangle\pm \sqrt{1-\lambda}|\eta_2\rangle$ and $|v_2\rangle=\sqrt{\lambda}|\eta_1\rangle\mp\sqrt{1-\lambda}|\eta_2\rangle$. From Eq.~\eqref{eq:expr_sigma}, it implies $a=b$. Hence, the assumption that $a\neq b$, implies $X_0\notin \ker \Psi$. Since $\ker\Phi^*$ has dimension $2$, so does $\ker \Phi=\ker \Psi$, where the equality follows from the definition of $\Psi$. Hence, commutation between $\Psi$ and $g(\Delta)$ implies $\ker\Phi=\ker \Psi=E_\lambda$. In other words,
$0=\Phi(|\eta_1\rangle\langle \eta_2|).$

That is equivalent to $\langle v_1,\eta_1\rangle\langle \eta_2, v_1\rangle=0$ and $\langle v_2,\eta_1\rangle\langle \eta_2, v_2\rangle=0$. The first equation implies $v_1$ is collinear with either $\eta_1$ or $\eta_2$. In both cases the second equation is equivalent to $\langle v_2,v_1\rangle\|v_2\|_2^2=0$, hence $v_1\perp v_2$ which contradicts our assumptions on $(v_1,v_2)$. Hence, $\Psi$ is KMS symmetric but does not commute with $g(\Delta)$, therefore, following Lemma~\ref{lem:BKMKMS_comute}, $\Psi$ is not BKM symmetric. We thus deduce that the QMS generator $\mathcal L=\Psi-\id$ verifies KMS QDB but not BKM QDB.

Now let $\widetilde{\Psi}=\mathcal M^{-1}\circ\Psi^*\circ\Gamma$. Since $\mathcal M^{-1}$ is completely positive and maps $\sigma$ to $\id$ and $\Gamma$ is completely positive and maps $\id$ to $\sigma$, $\widetilde{\Psi}$ is a quantum channel. Since $\Psi$ is KMS symmetric, using Eqs.~\eqref{eq:BKM_intertwine} and~\eqref{eq:KMS_intertwine}, it follows that $\widetilde{\Psi}$ is BKM symmetric. Since $\mathcal M^{-1}$ and $\Gamma$ both commute with $g(\Delta)$, $\widetilde{\Psi}$ is KMS symetric and therefore $\widetilde{\mathcal L}=\widetilde{\Psi}- \id$ verifies KMS QDB only if $\Psi$ commutes with $g(\Delta)$. That is not the case. Hence $\widetilde{\mathcal L}$ does verify BKM QDB but not KMS QDB.

Summarizing, we constructed QMS generators that verifiy either BKM QDB or KMS QDB but not both at the same time. Hence BKM QDB and KMS QDB properties are not comparable. We conclude with an example of a generator that verifies both KMS and BKM QDB but not GNS QDB.

This example is extracted from \cite{BCJPP}.  Take $K_1=\begin{pmatrix} \sqrt{p}&0\\0&\sqrt{1-p}\end{pmatrix}$ and $K_2=\begin{pmatrix} 0 &\sqrt{p}\\ \sqrt{1-p}& 0\end{pmatrix}$ with $p\in ]0,1/2[$. Let $\Phi:X\mapsto K_1^*XK_1+K_2^*XK_2$. It is an irreducible quantum channel. The density matrix $\sigma=\begin{pmatrix}p&0\\0&1-p\end{pmatrix}$ is its unique invariant state and direct computation shows that $\Phi$ is both KMS and BKM symmetric. However, denoting $\Phi^{\operatorname{GNS}}$ the GNS dual of $\Phi$, setting $P=\begin{pmatrix} 1 & 1\\1 &1\end{pmatrix}$ and $x=\begin{pmatrix}1 \\ -1\end{pmatrix}$, $\langle x,\Phi^{\operatorname{GNS}}(P)x\rangle\leq 0$ although $P\geq 0$, therefore $\Phi^{\operatorname{GNS}}$ is not positive, hence $\mathcal L= \Phi-\id$ does not verify GNS QDB.

\bibliographystyle{abbrv}
\bibliography{library}

\end{document}